\documentclass[journal]{IEEEtran}

\IEEEoverridecommandlockouts
\usepackage{cite}
\usepackage{amsmath,amssymb,amsfonts}
\usepackage{algorithmic}
\usepackage{graphicx}
\usepackage{textcomp}
\usepackage{xcolor}
\usepackage{algorithm}

\begin{document}

\title{Cooperative LEO-Terrestrial Multistatic ISAC: CRLB Analysis, Scaling Laws, and Satellite Selection}
 
\author{Yunhui~Li,~\IEEEmembership{Student Member,~IEEE}, Kaitao~Meng,~\IEEEmembership{Member,~IEEE}, Emad Alsusa,~\IEEEmembership{Senior Member,~IEEE} and Kaiting You,~\IEEEmembership{Student Member,~IEEE}

\thanks{Y.~Li, K.~Meng, E.~Alsusa and K. You are with the Department of Electrical and Electronic Engineering, The University of Manchester, Manchester, United Kingdom (e-mail: yunhui.li@manchester.ac.uk; kaitao.meng@manchester.ac.uk; e.alsusa@manchester.ac.uk; \\kaiting.you@manchester.ac.uk).}
}

% \markboth{Journal of \LaTeX\ Class Files,~Vol.~14, No.~8, August~2015}%
% {Shell \MakeLowercase{\textit{et al.}}: Bare Demo of IEEEtran.cls for IEEE Journals}

\maketitle
\begin{abstract}
Low Earth orbit (LEO) satellites provide elevated and spatially diverse viewpoints for enhancing three-dimensional (3-D) sensing in integrated satellite-terrestrial networks (ISTNs). This paper investigates a LEO-assisted terrestrial multistatic integrated sensing and communication (ISAC) network for 3-D target localisation, where multiple LEO satellites act as cooperative sensing illuminators and provide additional bistatic observations to distributed terrestrial radar receivers. To characterise the localisation performance, we first model cooperative satellites as a homogeneous Poisson point process (PPP) and derive a tractable approximation of the average hybrid Cramér–Rao lower bound (CRLB). The corresponding scaling laws show that the root-CRLB decreases proportionally to the inverse square root of the mean number of cooperative satellites when the cooperation region is fixed, whereas enlarging the cooperation radius at a fixed satellite density yields logarithmic diminishing returns. The analysis further reveals that the asymptotic vertical localisation gain depends critically on the terrestrial sensing geometry. To capture practical constellation characteristics, an Earth-curvature-aware Walker model is subsequently developed by incorporating orbital structure, satellite motion, visibility, and time-varying sensing geometry. A tractable approximation to the Walker-based hybrid CRLB is then derived. Analytical bounds on the marginal gain and a sufficient condition for ordering candidate satellites are also derived. Based on this analysis, a CRLB-oriented greedy satellite-selection strategy is proposed to jointly account for signal-to-clutter-plus-noise ratio (SCNR)-dependent reliability and geometry information complementarity with the existing terrestrial sensing configuration. The proposed selection strategy consistently outperforms benchmarks and achieves performance close to exhaustive search with substantially reduced selection complexity. The impact of terrestrial sensing information on satellite selection is also examined through satellite-only and hybrid-aware strategies. Monte Carlo simulations validate the analytical approximations for both models.
\end{abstract}


\section{Introduction}
\IEEEPARstart{I}{ntegrated} sensing and communication (ISAC) has emerged as a promising paradigm for next-generation networks\cite{b25,b26}. It is anticipated to be a key functionality for future wireless networks, enabling communication services and environmental sensing to be supported through shared wireless resources\cite{b27}. For target localisation, terrestrial multistatic sensing can exploit spatially distributed radar receivers to acquire multiple bistatic observations, with the resulting localisation accuracy being fundamentally determined by the sensing geometry\cite{b28,b3}. For low-altitude aerial targets such as unmanned aerial vehicles (UAVs), terrestrial ISAC base stations (BSs) are typically configured with down-tilted beams to serve ground users. Consequently, an aerial target may lie outside the main lobe of the terrestrial sensing beam and be illuminated only through weak sidelobes, resulting in degraded sensing signal-to-clutter-plus-noise ratio (SCNR) and limited elevation-domain information\cite{b29}.

These limitations motivate the integration of complementary non-terrestrial sensing platforms into terrestrial ISAC networks. Low Earth Orbit (LEO) satellite networks are emerging as an important component of future wireless systems to support global coverage and enable ubiquitous connectivity\cite{b1,b14}. Satellite networks can provide robust communication services in areas where terrestrial deployment is insufficient, such as rural regions and disaster-affected areas\cite{b2}. By integrating satellite and terrestrial segments, integrated satellite-terrestrial networks (ISTNs) combine broad-area coverage with terrestrial network capacity. Beyond these communication benefits, the elevated and mobile viewpoints of LEO satellites provide spatially diverse observation geometries that are difficult to obtain from terrestrial infrastructure alone\cite{b23}.  Compared with higher-orbit systems, LEO links can offer lower propagation delay and path loss. Therefore, LEO systems have attracted interest from industry\cite{b24}. 

From an aerial-sensing perspective, LEO satellites employ Earth-pointing beams and observe the target from above, allowing aerial targets within their coverage footprints to be directly illuminated from elevated directions. The resulting satellite-target-receiver paths therefore provide two complementary benefits: they extend the effective illumination coverage beyond that of the terrestrial beams and introduce high-elevation bistatic observations that strengthen the otherwise weak vertical localisation geometry. These advantages motivate the incorporation of LEO satellites for reliable 3-D localisation. Global Navigation Satellite System (GNSS)-based multistatic radar studies have reported that signals from more than 30 GNSS satellites may be simultaneously visible from a terrestrial location\cite{b4}, illustrating the potential geometric diversity of satellite-borne illuminators. Recent studies have also demonstrated the feasibility of multiple-satellite cooperative sensing in LEO constellations. In particular, a dual-function LEO satellite constellation framework was proposed in~\cite{b5}, where multiple LEO satellites cooperate to provide communication services for multiple user devices and location sensing for a target of interest using the same spectrum resources. This motivates the use of cooperative LEO satellites as distributed sensing illuminators for enhancing the geometric diversity of satellite-assisted target localisation.

Stochastic geometry provides a tractable means of characterising uncertainty in satellite availability and spatial geometry. Early studies modelled LEO satellites as a binomial point process (BPP) on spherical surfaces and derived coverage probability for satellite communication systems~\cite{b8,b11,b12}. Conventional models represent satellites on a spherical surface, which provides tractable spatial statistics but does not explicitly preserve the orbital structure of practical satellite constellations. To address this limitation, Cox point process models have been proposed to jointly model both the distribution of orbits and the satellites positioned along them, thereby avoiding the fully random spatial placement assumption used in conventional point process models~\cite{b9,b13}.  Moreover, nonhomogeneous stochastic geometry models were further developed to capture the latitude-dependent distribution of inclined LEO constellations, where the satellite intensity varies with the actual spatial distribution induced by constellation altitude, size, and orbital inclination~\cite{b6}. Walker constellation models have also been studied from stochastic-geometry and dynamical-system perspectives, where satellite orbits are regularly spaced and satellites are periodically distributed along the orbits~\cite{b10,b20}. These orbit-aware models are more consistent with practical satellite constellations and can capture the time-varying geometry induced by orbital motion. However, existing analyses mainly focus on communication metrics such as coverage probability and interference distributions rather than satellite-assisted sensing and localisation. The impact of orbit-constrained LEO satellite geometry on the Fisher information matrix (FIM) and Cramér-Rao lower bound (CRLB) of satellite-assisted terrestrial multistatic sensing remains insufficiently explored.

In this paper, we investigate a LEO-assisted terrestrial multistatic ISAC network for 3-D target localisation, where multiple LEO satellites cooperate with the terrestrial sensing infrastructure by illuminating the targets and providing additional bistatic observations. These satellite-assisted sensing paths enrich the spatial diversity of the terrestrial multistatic system and improve the hybrid localisation capability. We analyse the proposed framework using both stochastic Poisson point process (PPP) and Earth-curvature-aware Walker constellation models. The local PPP model is analytically convenient for deriving tractable CRLB approximations and explicit scaling laws, but it abstracts the orbital-plane structure and satellite motion. In contrast, the Walker model preserves the orbit-constrained and time-varying sensing geometry, although it is less amenable to closed-form analysis. The main contributions of this work are summarised as follows:
\begin{itemize}

\item We establish a stochastic-geometry framework for LEO-assisted terrestrial multistatic localisation by modelling the ground projections of cooperative satellites as a homogeneous PPP. We derive a tractable closed-form approximation of the PPP-averaged hybrid CRLB. Moreover, we derive scaling laws that characterise how the 3-D localisation accuracy evolves with the mean number of cooperative satellites $\bar K$ and the satellite cooperation radius. For a fixed cooperation region, the root-CRLB scales as $\mathcal O(\bar K^{-1/2})$. For a fixed satellite density, enlarging the cooperation radius yields only an inverse-square-root logarithmic improvement. The analysis further reveals that sustained vertical information growth depends on the terrestrial sensing geometry, and that increasing the density of nearby cooperative satellites provides a faster asymptotic localisation gain than extending the cooperation region to include increasingly distant satellites.

\item We develop an Earth-curvature-aware Walker-constellation model to bridge the tractable local stochastic analysis with orbit-constrained LEO deployments. The model captures orbital-plane structure, Walker phasing, satellite motion, elevation-based visibility, and time-varying sensing geometry. We further derive a tractable approximation of the Walker-based hybrid CRLB. 

\item We propose a CRLB-oriented greedy satellite selection strategy that jointly accounts for SCNR-based reliability and geometric information complementarity. To characterise the selection method, we derive the exact marginal trace-CRLB reduction of each candidate satellite and analyse it through the eigenstructure of the current hybrid FIM. The resulting analysis establishes how directional complementarity and information saturation affect the marginal localisation gain, and further provides analytical bounds and a sufficient condition for ordering candidate satellites. The greedy method aims to minimise the hybrid CRLB, and it achieves nearly the same performance as exhaustive search while reducing computational complexity. We further investigate the role of terrestrial-satellite information complementarity by distinguishing satellite-only selection from hybrid-aware selection.

\end{itemize}
\begin{figure}
\centering
\includegraphics[width=0.7\linewidth]{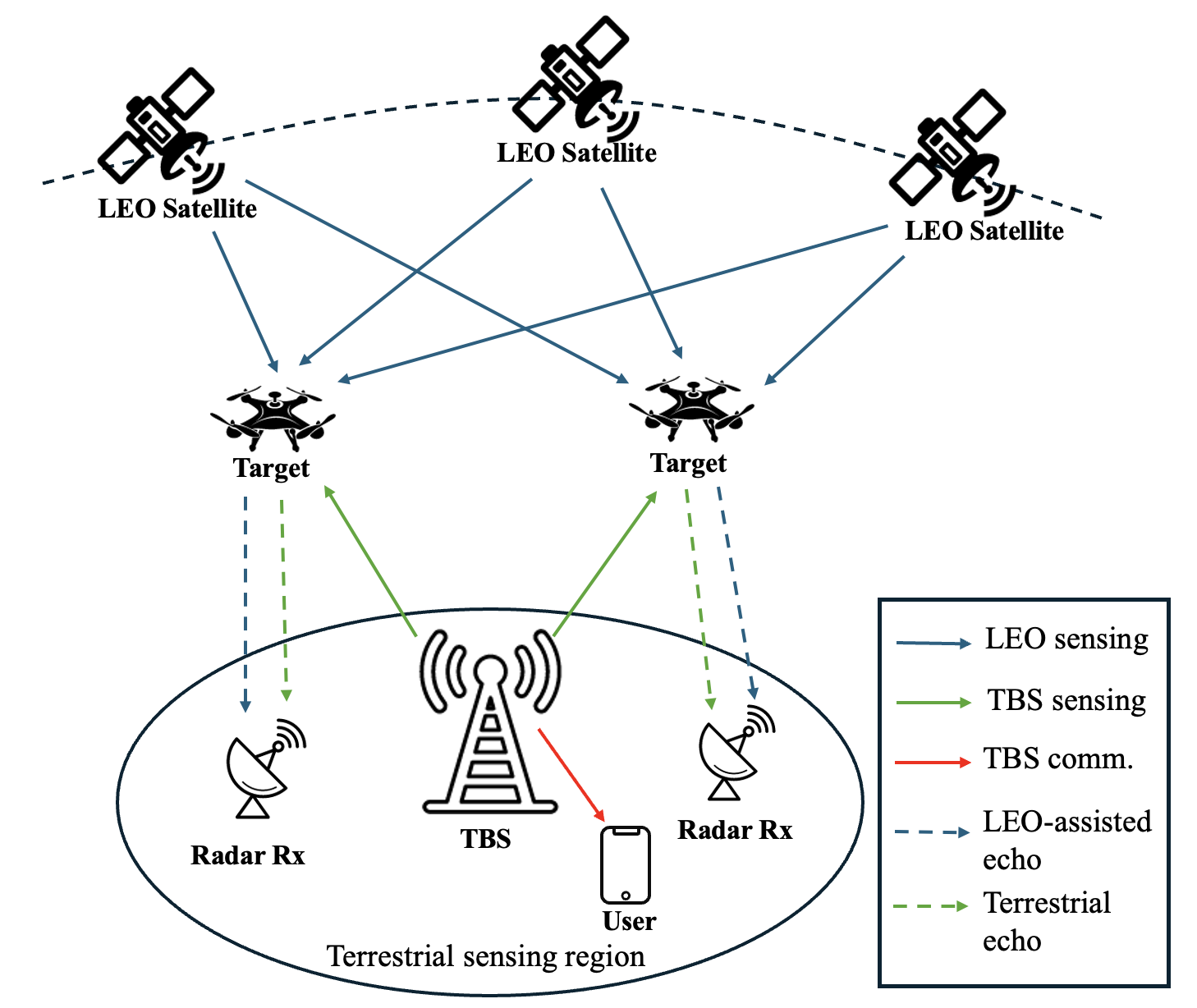} 
\caption{Illustration of the cooperative LEO-terrestrial multistatic ISAC system.}
\vspace{-1em}
\end{figure}
\vspace{-1em}
\section{System Model}
As shown in Fig. 1, we consider a satellite-assisted terrestrial multistatic ISAC network for 3-D target localisation. The terrestrial and satellite sensing branches are assumed to share the same frequency spectrum. The considered system consists of one terrestrial base station (TBS), a set of LEO satellites, multiple UAV targets, and multiple distributed radar receivers. The ISAC-enabled terrestrial network supports communication users while simultaneously transmitting radar signals for sensing targets. The TBS, each radar receiver, and each LEO satellite are equipped with $N_{\mathrm{TX}}$, $N_{\mathrm{RX}}$, and $N_{\mathrm{SAT}}$ antennas, respectively, where the radar receivers collect the echoes reflected by the targets and clutter scatterers. The array configurations of the TBS, radar receivers and LEO satellites are all assumed to be uniform rectangular arrays (URAs). The LEO satellites transmit sensing signals towards the targets. The target-reflected echoes generated by satellite illumination are collected by the distributed radar receivers, thereby forming additional bistatic sensing paths. Each satellite provides additional bistatic range observations through the distributed radar receivers for the targets, which enriches the observation geometry for 3-D localisation. After compensation for satellite-induced Doppler shifts and timing offsets, the terrestrial and satellite sensing signals are designed as mutually orthogonal waveforms. Therefore, the reflected echoes generated by the terrestrial and satellite sensing branches can be separated at the radar receivers.
\vspace{-1.2em}
\subsection{Channel Model}
Let $\mathbf{p}_{\mathrm{T}} \in \mathbb{R}^{3}$ denote the TBS position, $\mathbf{p}_{s} \in \mathbb{R}^{3}$ denote the position of the $s$th LEO satellite, $\mathbf{p}_{i} \in \mathbb{R}^{3}$ denote the position of the $i$th target, and $\mathbf{r}_{n} \in \mathbb{R}^{3}$ denote the position of the $n$th radar receiver. The position of the $c$th clutter scatterer around the targets is denoted by $\mathbf q_c\in\mathbb R^3$. The sets of satellites, sensing targets, radar receivers, and clutter scatterers are denoted by $\mathcal{S}=\{1,\ldots,N_{\mathrm{s}}\}$, $\mathcal{I}=\{1,\ldots,\mathrm{N_{tar}}\}$, $\mathcal{D}=\{1,\ldots,N_{\mathrm{r}}\}$, and $\mathcal{C}=\{1,\ldots,N_{\mathrm{cl}}\}$, respectively, where \(s\in\mathcal S\), \(i\in\mathcal I\), \(n\in\mathcal D\), and \(c\in\mathcal C\). Let $\mathbf a_{\mathrm T}(\boldsymbol{\theta})$, $\mathbf a_{\mathrm S}(\boldsymbol{\theta})$, and
$\mathbf a_{\mathrm R}(\boldsymbol{\theta})$ denote the array response vectors of the TBS, the satellite, and the radar receiver, respectively, where $\boldsymbol{\theta}$ denotes the corresponding azimuth-elevation direction. For the terrestrial sensing branch, the channel from the TBS to target $i$ and radar receiver $n$ is
\begin{equation}
    \mathbf{G}^{\mathrm{T,tar}}_{i,n}
    =
    \alpha^{\mathrm{T,tar}}_{i,n}
    \mathbf{a}_{\mathrm{R}}\!\left(\theta^{\mathrm{R}}_{i,n}\right)
    \mathbf{a}^{H}_{\mathrm{T}}\!\left(\theta^{\mathrm{T}}_{i}\right),
\end{equation}
where $\theta^{\mathrm{T}}_{i}$ is the direction from the TBS to the target $i$, $\theta^{\mathrm{R}}_{i,n}$ is the direction from the target $i$ to the radar receiver $n$, and $\alpha^{\mathrm{T,tar}}_{i,n}$ denotes the amplitude gain of the bistatic path from the TBS to the target $i$ and then to the radar receiver $n$, which includes path loss and radar cross-section (RCS) gain, and is modelled as\cite{b15}
\begin{equation}
    \alpha^{\mathrm{T,tar}}_{i,n}
    =
    \sqrt{
    \frac{
    \lambda_c^{2}\sigma_{i,n}^{2}G_{\mathrm{T}}^{\mathrm{v}}(\mathbf{p}_{i})
    }
    {
    (4\pi)^{3}d_{\mathrm{T},i}^{2}d_{i,n}^{2}
    }},
\end{equation}
where $ \lambda_c$, $\sigma_{i,n}^2$, $d_{\mathrm{T},i}^{2}$, and $d_{i,n}^{2}$ are the wavelength, RCS gain, distance from the TBS to the target $i$, and the distance from the target $i$ to the receiver $n$, respectively. The TBS employs a down-tilted vertical radiation pattern. $G^{\mathrm{v}}_{\mathrm{T}}(\mathbf{p}_{i})$ is the elevation-dependent TBS transmit gain towards target $i$. 

To represent practical multistatic ISAC environments, the additional echoes from unintended targets and clutter are considered. The terrestrial clutter channel from the TBS to clutter scatterer $c$ and then to radar receiver $n$ is expressed as
\begin{equation}
    \mathbf{G}^{\mathrm{T,cl}}_{c,n}
    =
    \alpha^{\mathrm{T,cl}}_{c,n}
    \mathbf{a}_{\mathrm{R}}\!\left(\theta^{\mathrm{R}}_{c,n}\right)
    \mathbf{a}^{H}_{\mathrm{T}}\!\left(\theta^{\mathrm{T}}_{c}\right),
\end{equation}
where $\theta^{\mathrm{T}}_{c}$ and $\theta^{\mathrm{R}}_{c,n}$ denote the corresponding transmit and receive directions, and $\alpha^{\mathrm{T,cl}}_{c,n}$ is the path gain from the clutter scatterer $c$  to the radar receiver $n$.

For the satellite-assisted sensing branch, the channel from satellite $s$ to target $i$ and radar receiver $n$ is expressed as
\vspace{-0.1em}
\begin{equation}
    \mathbf{H}^{\mathrm{S,tar}}_{s,i,n}
    =
    \alpha^{\mathrm{S,tar}}_{s,i,n}
    \mathbf{a}_{\mathrm{R}}\!\left(\theta^{\mathrm{R}}_{i,n}\right)
    \mathbf{a}^{H}_{\mathrm{S}}\!\left(\theta^{\mathrm{S}}_{s,i}\right),
\end{equation}
\vspace{-1em}

\noindent where $\theta^{\mathrm{S}}_{s,i}$ is the direction from satellite $s$ to target $i$. The coefficient $\alpha^{\mathrm{S,tar}}_{s,i,n}$ denotes the path gain associated with the satellite-assisted sensing path from satellite $s$ to target $i$ and subsequently to radar receiver $n$.

The satellite-assisted clutter channel from satellite $s$ to clutter scatterer $c$ and then to radar receiver $n$ is given by
\vspace{-0.4em}
\begin{equation}
    \mathbf{H}^{\mathrm{S,cl}}_{s,c,n}
    =
    \alpha^{\mathrm{S,cl}}_{s,c,n}
    \mathbf{a}_{\mathrm{R}}\!\left(\theta^{\mathrm{R}}_{c,n}\right)
    \mathbf{a}^{H}_{\mathrm{S}}\!\left(\theta^{\mathrm{S}}_{s,c}\right),
\end{equation} 
\vspace{-1.em}

\noindent where $\theta^{\mathrm{S}}_{s,c}$ is the direction from satellite $s$ to the clutter scatterer $c$ and the coefficient $\alpha^{\mathrm{S,cl}}_{s,c,n}$ is the path gain from satellite $s$ to the clutter scatterer $c$ and then to the radar receiver $n$. The remaining path-gain coefficients $\alpha^{\mathrm{T,cl}}_{c,n}$, $\alpha^{\mathrm{S,tar}}_{s,i,n}$, and $\alpha^{\mathrm{S,cl}}_{s,c,n}$ follow the same bistatic radar-equation form as (2), evaluated using the corresponding transmitter--scatterer and scatterer--receiver distances and the associated transmit and receive antenna gains.

\vspace{-1.2em}
\subsection{Signal Model}
Let $\mathcal{A}=\{1,\ldots,N_{\mathrm{com}}\}$ denote the set of terrestrial communication users. The ISAC-enabled TBS simultaneously transmits communication signals to terrestrial users and sensing signals towards the targets. The signal vector of the TBS is defined as $\mathbf{s}_{\mathrm{ter}}=[s^{\mathrm{c}}_{1},\ldots,s^{\mathrm{c}}_{N_{\mathrm{com}}},s^{\mathrm{ter}}_{1},\ldots,s^{\mathrm{ter}}_{N_{\mathrm{tar}}}]^{T}$, where $s^{\mathrm{c}}_{a}$ is the data stream for communication user $a$, and $s^{\mathrm{ter}}_{i}$ is the radar signal for target $i$. The corresponding beamforming matrix is denoted by $\mathbf{F}_{\mathrm{ter}}=[\mathbf{f}^{\mathrm{c}}_{1},\ldots,\mathbf{f}^{\mathrm{c}}_{N_{\mathrm{com}}},\mathbf{f}^{\mathrm{ter}}_{1},\ldots,\mathbf{f}^{\mathrm{ter}}_{N_{\mathrm{tar}}}]$, where $\mathbf{f}^{\mathrm{c}}_{a}$ and $\mathbf{f}^{\mathrm{ter}}_{i}$ are the beamforming vectors for user $a$ and target $i$, respectively. For satellite-assisted sensing, satellite $s$ transmits dedicated sensing signals towards the targets. The signal vector of satellite $s$ is defined as $\mathbf{s}^{\mathrm{sat}}_{s}=[s^{\mathrm{sat}}_{s,1},\ldots,s^{\mathrm{sat}}_{s,N_{\mathrm{tar}}}]^{T}$, where $s^{\mathrm{sat}}_{s,i}$ is the sensing signal transmitted by satellite $s$ for target $i$. The satellite sensing beamforming matrix is denoted by $\mathbf{F}^{\mathrm{sat}}_{s}=[\mathbf{f}^{\mathrm{sat}}_{s,1},\ldots,\mathbf{f}^{\mathrm{sat}}_{s,N_{\mathrm{tar}}}]$, where $\mathbf{f}^{\mathrm{sat}}_{s,i}$ is the beamforming vector used by satellite $s$ for target $i$. The beamforming vectors are normalised before power allocation. The TBS has a total transmit power $P_{\mathrm{TBS}}$, which is equally allocated among its communication and sensing beams, while the total transmit power $P_{\mathrm{sat}}$ of each LEO satellite is equally allocated among its sensing beams. The received signal of the terrestrial sensing branch at radar receiver $n$ is modelled as
\vspace{-0.3em}
\begin{equation}
\begin{aligned}
\mathbf{y}^{\mathrm{ter}}_{n}(t)
={}&
\sum_{i\in\mathcal{I}}
e^{-j\omega_c\tau^{\mathrm{ter}}_{i,n}}
\mathbf{G}^{\mathrm{T,tar}}_{i,n}
\mathbf{F}_{\mathrm{ter}}
\mathbf{s}_{\mathrm{ter}}
\left(t-\tau^{\mathrm{ter}}_{i,n}\right)
\\
&+
\sum_{c\in\mathcal{C}}
e^{-j\omega_c\tau^{\mathrm{ter}}_{c,n}}
\mathbf{G}^{\mathrm{T,cl}}_{c,n}
\mathbf{F}_{\mathrm{ter}}
\mathbf{s}_{\mathrm{ter}}
\left(t-\tau^{\mathrm{ter}}_{c,n}\right)
+
\mathbf{n}^{\mathrm{ter}}_{n}(t).
\end{aligned}
\end{equation}
where $\mathbf{y}^{\mathrm{ter}}_{n}$ is the received terrestrial-branch signal vector, $\mathbf{n}^{\mathrm{ter}}_{n}$ is the additive white Gaussian noise (AWGN) vector with variance $\sigma_{\mathrm{ter}}^{2}$ and $\omega_c=2\pi f_c$ is the carrier angular frequency. $\tau^{\mathrm{ter}}_{i,n}$ and $\tau^{\mathrm{ter}}_{c,n}$ denote the bistatic propagation delays along the TBS-target-receiver and TBS-clutter-receiver paths, respectively.

The received signal of the satellite-assisted sensing branch associated with satellite $s$ and radar receiver $n$ is modelled as
\begin{equation}
\begin{aligned}
\mathbf{y}^{\mathrm{sat}}_{s,n}(t)
={}&
\sum_{i\in\mathcal{I}}
e^{-j\omega_c\tau^{\mathrm{sat}}_{s,i,n}}
\mathbf{H}^{\mathrm{S,tar}}_{s,i,n}
\mathbf{F}^{\mathrm{sat}}_{s}
\mathbf{s}^{\mathrm{sat}}_{s}
\left(t-\tau^{\mathrm{sat}}_{s,i,n}\right)
\\
&\hspace{-1.9em}+
\sum_{c\in\mathcal{C}}
e^{-j\omega_c\tau^{\mathrm{sat}}_{s,c,n}}
\mathbf{H}^{\mathrm{S,cl}}_{s,c,n}
\mathbf{F}^{\mathrm{sat}}_{s}
\mathbf{s}^{\mathrm{sat}}_{s}
\left(t-\tau^{\mathrm{sat}}_{s,c,n}\right)
+
\mathbf{n}^{\mathrm{sat}}_{s,n}(t).
\end{aligned}
\end{equation}
where $\mathbf{y}^{\mathrm{sat}}_{s,n}$ is the received satellite-branch signal vector, and $ \mathbf{n}^{\mathrm{sat}}_{s,n}$ is the AWGN vector with variance $\sigma_{\mathrm{sat}}^{2}$.  $\tau^{\mathrm{sat}}_{s,i,n}$ and $\tau^{\mathrm{sat}}_{s,c,n}$ denote the bistatic propagation delays along satellite-target-receiver and satellite-clutter-receiver paths, respectively.

The probing waveforms transmitted by different LEO satellites are assumed to be orthogonal. Therefore, inter-satellite sensing interference is not considered. The streams associated with different beamforming vectors are modelled as mutually uncorrelated. To improve sensing performance, each radar receiver employs the minimum variance distortionless response (MVDR) receive beamformer to preserve the desired target echo while suppressing interference and noise\cite{b7}. Then, the received signals for detecting the echo from the target $i$ are given by $\left(\mathbf{w}^{\mathrm{sat}}_{s,i,n}\right)^H \mathbf{y}^{\mathrm{sat}}_{s,n}$
and $\left(\mathbf{w}^{\mathrm{ter}}_{i,n}\right)^H
\mathbf{y}^{\mathrm{ter}}_{n}$ for the satellite-assisted and terrestrial sensing branches, respectively. For wireless sensing, the SCNR is adopted as the performance metric\cite{b16}. The satellite-assisted SCNR for target $i$ associated with satellite $s$ and radar receiver $n$ is given by
\vspace{-0.5em}
\begin{equation}
    \mathrm{SCNR}^{\mathrm{sat}}_{s,i,n}
    =
    \frac{
    \left|
    \mathbf{(w}^{\mathrm{sat}}_{s,i,n})^H
    \mathbf{H}^{\mathrm{S,tar}}_{s,i,n}
    \mathbf{f}^{\mathrm{sat}}_{s,i}
    \right|^{2}
    }{
    \sum_{k=1}^{3} I^{\mathrm{sat},(k)}_{s,i,n}
    +
    \sigma_{\mathrm{sat}}^{2}
    \left\|
    \mathbf{w}^{\mathrm{sat}}_{s,i,n}
    \right\|^{2}
    } ,
\end{equation}
where $I^{\mathrm{sat},(1)}_{s,i,n}
= \sum_{j\in\mathcal I,j\neq i} \left|
(\mathbf w^{\mathrm{sat}}_{s,i,n})^H
\mathbf H^{\mathrm{S,tar}}_{s,j,n}
\mathbf f^{\mathrm{sat}}_{s,i}
\right|^2$,
$I^{\mathrm{sat},(2)}_{s,i,n}
=
\sum_{c\in\mathcal C}
\left\|
(\mathbf w^{\mathrm{sat}}_{s,i,n})^H
\mathbf H^{\mathrm{S,cl}}_{s,c,n}
\mathbf F^{\mathrm{sat}}_{s}
\right\|^2$,
and
$I^{\mathrm{sat},(3)}_{s,i,n}
=
\sum_{\substack{j\in\mathcal I\\j\neq i}}
\sum_{q\in\mathcal I}
\left|
(\mathbf w^{\mathrm{sat}}_{s,i,n})^H
\mathbf H^{\mathrm{S,tar}}_{s,q,n}
\mathbf f^{\mathrm{sat}}_{s,j}
\right|^2$ represent the interference from unintended targets, clutter, and other satellite sensing beams, respectively.

The terrestrial branch SCNR for target $i$ at radar receiver $n$ is defined as
\vspace{-0.45em}
\begin{equation}
    \mathrm{SCNR}^{\mathrm{ter}}_{i,n}
    =
    \frac{
    \left|
    \mathbf{(w}^{\mathrm{ter}}_{i,n})^H
    \mathbf{G}^{\mathrm{T,tar}}_{i,n}
      \mathbf{f}^{\mathrm{ter}}_{i}
    \right|^{2}
    }{
    \sum_{k=1}^{3} I^{\mathrm{ter},(k)}_{i,n}
    +
    \sigma_{\mathrm{ter}}^{2}
    \left\|
    \mathbf{w}^{\mathrm{ter}}_{i,n}
    \right\|^{2}
    } ,
\end{equation}
where
$I^{\mathrm{ter},(1)}_{i,n}
=
\sum_{j\in\mathcal{I},j\neq i}
\left|
\left(\mathbf{w}^{\mathrm{ter}}_{i,n}\right)^H
\mathbf{G}^{\mathrm{T,tar}}_{j,n}
\mathbf{f}^{\mathrm{ter}}_{i}
\right|^2$,
$I^{\mathrm{ter},(2)}_{i,n}
=
\sum_{c\in\mathcal{C}}
\left\|
\left(\mathbf{w}^{\mathrm{ter}}_{i,n}\right)^H
\mathbf{G}^{\mathrm{T,cl}}_{c,n}
\mathbf{F}_{\mathrm{ter}}
\right\|^2$,
and
$I^{\mathrm{ter},(3)}_{i,n}
=
\sum_{a\in\mathcal{A}}
\sum_{q\in\mathcal{I}}
\left|
\left(\mathbf{w}^{\mathrm{ter}}_{i,n}\right)^H
\mathbf{G}^{\mathrm{T,tar}}_{q,n}
\mathbf{f}^{c}_{a}
\right|^2
+
\sum_{\substack{j\in\mathcal{I}\\ j\neq i}}
\sum_{q\in\mathcal{I}}
\left|
\left(\mathbf{w}^{\mathrm{ter}}_{i,n}\right)^H
\mathbf{G}^{\mathrm{T,tar}}_{q,n}
\mathbf{f}^{\mathrm{ter}}_{j}
\right|^2$ represent the interference from unintended targets, clutter, and other TBS beamforming vectors, respectively. In addition, $\mathbf{w}^{\mathrm{ter}}_{i,n}$ is the terrestrial MVDR receive beamforming vector.

\vspace{-0.3em}
\section{Sensing Performance Analysis}
This section develops a PPP-based stochastic geometry analysis and explicit scaling laws, introducing an Earth-curvature-aware time-varying Walker model and developing a theoretically characterised CRLB-oriented satellite-selection framework. The CRLB analysis considers a representative target at the centre of the sensing region. 
\vspace{-1em}
\subsection{PPP-Based Stochastic Satellite Model}
The TBS position is represented as $\mathbf p_{\mathrm T}=[0,0,h_{\mathrm T}]^{T}$ and the target position is denoted by $\mathbf p_i= [0,0,h_t]^{\mathrm{T}}.$ The terrestrial radar receiver $n$ is fixed at
\begin{equation}
    \mathbf r_n =
    [R_r\cos\theta_n,\ R_r\sin\theta_n,\ h_r]^{\mathrm{T}},
    \quad
    \theta_n=\frac{2\pi(n-1)}{N_{\mathrm{r}}}.
\end{equation}
where $R_r$ is the receiver sensing radius. 

The ground projections of the cooperative satellites are modelled as a homogeneous PPP \(\Phi_{\mathrm{s}}\) with density \(\lambda_s\) over the local cooperation region
\begin{equation}
    \mathcal B(D_{\max})
    =
    \left\{
    (\rho,\phi_s):0\leq \rho\leq D_{\max},\;0\leq \phi_s<2\pi
    \right\}.
\end{equation}
where \(\lambda_s\) denotes the average satellite density, and \(D_{\max}\) is the radius of the cooperation region. The number of satellites within the disk is Poisson distributed with mean $\lambda_s\pi D_{\max}^2$, and conditioned on this number, their ground-projection locations are i.i.d. uniform, with $f_{\rho}(\rho)=\frac{2\rho}{D_{\max}^2}$ and $f_{\phi_s}(\phi_s)=\frac{1}{2\pi}$.

A satellite with polar coordinates $(\rho,\phi_{s})$ can be expressed as $\mathbf p_s =[\rho\cos\phi_{s},\ \rho\sin\phi_{s},\ H]^{\mathrm{T}}$, where $\rho$ denotes the horizontal distance from the centre target to the ground projection of satellite $s$, $H$ is the satellite altitude, and $\phi_s$ denotes the corresponding azimuth angle. Therefore, the corresponding distance between the satellite and the target is given by $d_s(\rho)=\|\mathbf p_s-\mathbf p_i\| = \sqrt{\rho^2+h^2}$, where $h = H-h_t$ denotes the vertical distance between the satellite and the target. Moreover, the distance between the target and receiver is $d_r = \sqrt{R_r^2+(h_t-h_r)^2}$ and let $c_r=\frac{R_r}{d_r},b_r=\frac{h_t-h_r}{d_r}$. For the satellite-assisted bistatic observation generated by satellite $\mathbf p_s$ and receiver $\mathbf r_n$, the range gradient vector is
\vspace{-0.3em}
\begin{equation}
    \mathbf g_{s,n}
    =
    \frac{\mathbf p_i-\mathbf p_s}{\|\mathbf p_i-\mathbf p_s\|}
    +
    \frac{\mathbf p_i-\mathbf r_n}{\|\mathbf p_i-\mathbf r_n\|} =  \begin{bmatrix}
    -a_s\cos\phi_{s}-c_r\cos\theta_n\\
    -a_s\sin\phi_{s}-c_r\sin\theta_n\\
    b_r-b_s
    \end{bmatrix}.
\end{equation}
where $a_s(\rho)=\frac{\rho}{d_s(\rho)}$and $b_s(\rho)=\frac{h}{d_s(\rho)}$.

For analytical tractability, following the satellite-assisted SCNR formulation in Section~II-B, we adopt a simplified noise-limited and range-dependent SCNR model~\cite{b1} only for the closed-form analysis. The resulting effective satellite-assisted SCNR is denoted by $\gamma_{\mathrm{eff}}^{\mathrm{sat}}(\rho)$. Therefore, the corresponding range information weight is expressed as
\begin{equation}
    \omega_{\rm eff}^{\rm sat}(\rho)
    =
    \frac{8\pi^2B^2}{c_0^2}
    \gamma_{\rm eff}^{\rm sat}(\rho)
    =
    \frac{K_{\rm eff}^{\rm sat}}{\rho^2+h^2},
    \label{eq:effective_sat_omega}
\end{equation}
where $B$ is the RMS bandwidth, $c_0$ is the speed of light, and $
    K_{\rm eff}^{\rm sat}
    =
    \frac{8\pi^2B^2}{c_0^2}
    \Gamma_{\rm eff}^{\rm sat}
    \label{eq:effective_sat_K}
$ with $ \Gamma_{\rm eff}^{\rm sat}
    =
    \frac{P_s G_{\rm eff}\Xi_{\rm sat}}{\sigma_n^2}$, where $P_s$ is the satellite sensing transmit power per beam, $G_{\rm eff}$ denotes the effective post-beamforming array gain, $\sigma_n^2$ is the receiver noise power, and $\Xi_{\rm sat}$ denotes the path gain after extracting the dominant satellite-to-target range-dependent term $1/(\rho^2+h^2)$ with $\Xi_{\rm sat}=  \frac{\lambda_c^2\sigma_{\rm tar}^2}{(4\pi)^3d_r^2}$, where \(\sigma_{\rm tar}^2\) is the average target RCS.

The terrestrial FIM is obtained  from bistatic measurements between the fixed TBS and the radar receivers as
\begin{equation}
    d_{\rm T}
    =
    \|\mathbf p_i-\mathbf p_{\rm T}\|,
    \quad
    \mathbf g_{{\rm ter},n}
    =
    \frac{\mathbf p_i-\mathbf p_{\rm T}}{\|\mathbf p_i-\mathbf p_{\rm T}\|}
    +
    \frac{\mathbf p_i-\mathbf r_n}{\|\mathbf p_i-\mathbf r_n\|},
\end{equation}

The bistatic paths are assumed to be delay-resolvable, and the unknown complex reflection coefficients are eliminated. Therefore, the terrestrial FIM is
$\mathbf J_{\rm T}
=\sum_{n=1}^{N_{\rm r}}\omega_{{\rm eff},n}^{\rm ter}
\mathbf g_{{\rm ter},n}\mathbf g_{{\rm ter},n}^{\rm T}$, while the satellite-assisted FIM for one satellite is
\vspace{-0.8em}
\begin{equation}
\mathbf J_s(\rho,\phi_s)
=
\sum_{n=1}^{N_{\rm r}}
\omega_{\rm eff}^{\rm sat}(\rho)
\mathbf g_{s,n}\mathbf g_{s,n}^{\rm T}.
\label{eq:one_satellite_fim}
\end{equation}
where $ \omega_{{\rm eff},n}^{\rm ter}$ is the terrestrial range information weight.

To simplify the subsequent derivation, we denote $\omega_{\rm eff}^{\rm sat}(\rho)$ by $\omega(\rho)$. Averaging over the satellite azimuth and radial-location distributions gives
\begin{equation}
    \mathbb E_{\rho,\phi_s}[\mathbf J_s]
    =
    \operatorname{diag}(\mu_x,\mu_x,\mu_z),
\end{equation}
where
$\mu_x=\frac{N_{\mathrm r}}{2}
\mathbb E_{\rho}\!\left[
\omega(\rho)\left(a_s^2(\rho)+c_r^2\right)
\right]$
and
$\mu_z=N_{\mathrm r}
\mathbb E_{\rho}\!\left[
\omega(\rho)\left(b_r-b_s(\rho)\right)^2
\right]$.

By applying Campbell's theorem\cite{b22}, the mean satellite-assisted FIM becomes $\bar K \mathbb E_{\rho,\phi_s}[\mathbf J_s]$, where \(\bar K\) denotes the mean number of cooperative satellites in the cooperation region. Owing to the centred target and the symmetric circular receiver deployment, the cross terms in $\mathbf J_{\rm T}$ cancel, yielding a diagonal terrestrial FIM. We denote its diagonal entries by $J_{{\rm T},x}=[\mathbf J_{\rm T}]_{1,1}$, $J_{{\rm T},y}=[\mathbf J_{\rm T}]_{2,2}$, and $J_{{\rm T},z}=[\mathbf J_{\rm T}]_{3,3}$. Therefore, the mean hybrid FIM can be expressed as
\begin{equation}
    \bar{\mathbf J}
    =
    \mathbf J_T
    +
    \bar K\mathrm{diag}(\mu_x,\mu_x,\mu_z)
    =
    \mathrm{diag}(\Lambda_x,\Lambda_y,\Lambda_z),
\end{equation}
where $\Lambda_x=J_{T,x}+\bar K\mu_x$, $\Lambda_y=J_{T,y}+\bar K\mu_x$, and $\Lambda_z=J_{T,z}+\bar K\mu_z$. 

The average CRLB is defined as $\overline{\mathrm{CRLB}} =
    \mathbb E_{\Phi_s}
    \left[
    \operatorname{tr}
    \left(
    \mathbf J^{-1}(\Phi_s)
    \right)
    \right]$. Let $\Delta\mathbf J=\mathbf J(\Phi_s)-\bar{\mathbf J}$. A second-order Taylor expansion of the inverse FIM gives
\begin{equation}
    \mathbf J^{-1}(\Phi_s)
    \approx
    \bar{\mathbf J}^{-1}
    -
    \bar{\mathbf J}^{-1}
    \Delta\mathbf J
    \bar{\mathbf J}^{-1}
    +
    \bar{\mathbf J}^{-1}
    \Delta\mathbf J
    \bar{\mathbf J}^{-1}
    \Delta\mathbf J
    \bar{\mathbf J}^{-1},
\end{equation}

When the random FIM fluctuation $\Delta\mathbf J$ is sufficiently small relative to the mean FIM $\bar{\mathbf J}$, the third- and higher-order terms in the inverse-FIM expansion can be neglected. Therefore, since $\mathbb E_{\Phi_s}[\Delta\mathbf J]=\mathbf 0$, the expected CRLB is approximated as
\begin{equation}
\overline{\mathrm{CRLB}}
\approx {}
\operatorname{tr}
\left(
\bar{\mathbf J}^{-1}
\right)
+
\mathbb E_{\Phi_s}
\left[
\operatorname{tr}
\left(
\bar{\mathbf J}^{-1}
\Delta\mathbf J
\bar{\mathbf J}^{-1}
\Delta\mathbf J
\bar{\mathbf J}^{-1}
\right)
\right].
\end{equation}

Let $A_x=\frac{1}{\Lambda_x},A_y=\frac{1}{\Lambda_y}$ and $A_z=\frac{1}{\Lambda_z}$, and define $\zeta_s(\rho)= b_r-b_s(\rho)$. The first-order moments are then evaluated as
\begin{align}
    \mu_x
    &=
    \frac{N_{\mathrm{r}}K_{\rm eff}^{\rm sat}}{2}
    \left[
    (1+c_r^2)\mathcal I(-1)
    -
    h^2\mathcal I(-2)
    \right],\\
    \mu_z
    &=
    N_{\mathrm{r}}K_{\rm eff}^{\rm sat}
    \left[
    b_r^2\mathcal I(-1)
    -
    2b_rh\mathcal I(-3/2)
    +
    h^2\mathcal I(-2)
    \right],
\end{align}
where the function is defined as
\vspace{-0.5em}
\begin{equation}
    \mathcal I(\nu)
    =
    \frac{1}{D_{\max}^{2}}
    \int_{h^{2}}^{h^{2}+D_{\max}^{2}}
    t^{\nu}\,{\rm d}t.
    \label{eq:app_auxiliary_integral}
\end{equation}

To evaluate the second-order Taylor correction, we require the second-order moments of FIM fluctuation. 
\newtheorem{proposition}{Proposition}
\begin{proposition}
Each FIM entry satisfies $\mathbb E_{\Phi_s}
    \left[
    \left(\Delta J_{ij}\right)^2
    \right]
    =
    \bar K Q_{ij}$ with $Q_{ij}
    =
    \mathbb E_{\rho,\phi_s}
    \left[
    J_{s,ij}^2(\rho,\phi_s)
    \right]
$.  Therefore, the second-order correction can be expressed as
\vspace{-0.4em}
\begin{align}
    \mathcal T_{\rm taylor}
    &=
    A_x^3Q_{xx}
    +A_y^3Q_{yy}
    +A_z^3Q_{zz}
    +(A_x^2A_y+A_y^2A_x)Q_{xy}
    \nonumber\\
    &\hspace{-1.2em}
    +(A_x^2A_z+A_z^2A_x)Q_{xz}
    +(A_y^2A_z+A_z^2A_y)Q_{yz}.
    \label{eq:T_shot_closed}
\end{align}
\end{proposition}

\begin{IEEEproof}
Please refer to Appendix A.
\end{IEEEproof}

Therefore, the PPP-averaged closed-form CRLB approximation can be expressed as
\begin{equation}
    \overline{\mathrm{CRLB}}
    \approx
    A_x+A_y+A_z
    +
    \bar K\mathcal T_{\rm taylor}.
    \label{eq:closed_form_crlb_final}
\end{equation}

Based on \eqref{eq:closed_form_crlb_final}, we further characterise the scaling law of the hybrid sensing accuracy with respect to the mean number of cooperative satellites. 
\newtheorem{theorem}{Theorem}
\begin{theorem}
For fixed $D_{\max}$ and $H$, and with the per-satellite transmit power and sensing bandwidth fixed, the PPP-averaged closed-form CRLB approximation follows
\begin{equation}
    \mathrm{CRLB}
     \sim    
    \frac{1}{\bar K}
    \left(
        \frac{2}{\mu_x}
        +
        \frac{1}{\mu_z}
    \right)
    +
    \mathcal O(\bar K^{-2}).
    \label{eq:trace_crlb_k_scaling}
\end{equation}

Therefore, the root-CRLB approximation decreases proportionally to  $1/\sqrt{ {\bar{K}}}$. 
\end{theorem}
\begin{IEEEproof}
Please refer to Appendix B.
\end{IEEEproof}

We next characterise the scaling law with respect to the satellite sensing radius $D_{\max}$ while keeping the satellite density $\lambda_s$ fixed. The satellite-assisted FIM contribution along the horizontal dimension is
\begin{equation}
\begin{aligned}
    \bar K(D_{\max})\mu_x(D_{\max})
    &=
    \frac{\lambda_s\pi N_{\mathrm{r}}K_{\mathrm{eff}}^{\mathrm{sat}}}{2}
    \Bigg[
    (1+c_r^2)
    \ln\left(
        1+\frac{D_{\max}^2}{h^2}
    \right) \\
    &\qquad
    -1+
    \frac{h^2}{h^2+D_{\max}^2}
    \Bigg].
\end{aligned}
\end{equation}
\newtheorem{lemma}{Lemma}
\begin{lemma}
For fixed $\lambda_s$, $H$, per-satellite transmit power and sensing bandwidth, the satellite-assisted FIM contribution along each horizontal dimension satisfies
\vspace{-0.05em}
\begin{equation}
\begin{aligned}
&\bar K(D_{\max})\mu_x(D_{\max}) \\
&\quad =
\begin{cases}
\displaystyle
\mathcal O\left(
\frac{D_{\max}^2}{h^2}
\right),
& D_{\max}\ll h,
\\[2mm]
\displaystyle
\mathcal O\left[
\ln\left(
1+\frac{D_{\max}^2}{h^2}
\right)
\right],
& D_{\max}\gg h.
\end{cases}
\end{aligned}
\end{equation}
\end{lemma}

\begin{IEEEproof}
Please refer to Appendix~C.
\end{IEEEproof}

The horizontal FIM therefore transitions from a quadratic-growth regime to a logarithmic-growth regime when $D_{\max}$ is of the same order as $h$. Moreover, the vertical satellite-assisted FIM contribution is given by
\begin{equation}
\begin{aligned}
\bar K(D_{\max})\mu_z(D_{\max})
&=
\lambda_s\pi N_{\mathrm r}K_{\mathrm{eff}}^{\mathrm{sat}}
\Bigg[
b_r^2\ln\left(1+\frac{D_{\max}^2}{h^2}\right) \\
&\hspace{-5em}
-4b_r\left(
1-\frac{1}{\sqrt{1+\frac{D_{\max}^2}{h^2}}}
\right)
+\frac{\frac{D_{\max}^2}{h^2}}
{1+\frac{D_{\max}^2}{h^2}}
\Bigg],
\end{aligned}
\end{equation}

\begin{lemma}
For fixed $\lambda_s$, $H$, and the per-satellite sensing parameters, the vertical satellite-assisted FIM contribution in the regime $D_{\max}\gg h$ satisfies
\begin{equation}
\bar K(D_{\max})\mu_z(D_{\max})
=
\begin{cases}
\displaystyle
\mathcal O\!\left[
\ln\!\left(
1+\frac{D_{\max}^2}{h^2}
\right)
\right],
& b_r\neq0,
\\[1mm]
\displaystyle
\mathcal O(1),
& b_r=0.
\end{cases}
\end{equation}
\end{lemma}

\begin{IEEEproof}
Please refer to Appendix~D.
\end{IEEEproof}

When $b_r=0$, the vertical satellite-assisted FIM converges to a finite value, resulting in a non-zero localisation floor. Considering Eqs.~(27) and (29) jointly, the sustained growth of the vertical FIM depends critically on the geometry of the terrestrial network. Combining the horizontal and vertical FIM contributions, the overall satellite-assisted FIM exhibits logarithmic growth with $D_{\max}\rightarrow\infty$ for $b_r\neq 0$ as $\mathcal O\!\left[\ln\!\left(1+D_{\max}^2/h^2\right)\right]$. Furthermore, as established in (25), increasing the mean number of cooperative satellites while keeping $D_{\max}$ fixed yields a root-CRLB scaling of $\mathcal O(\bar K^{-1/2})$. The large-radius analysis shows that increasing $\bar K$ by enlarging $D_{\max}$ while keeping $\lambda_s$ fixed yields a slower scaling of $\mathcal O((\ln\bar K)^{-1/2})$ for $b_r\neq0$. Hence, increasing the density of nearby cooperative satellites provides a faster asymptotic localisation gain than extending the cooperation region to include increasingly distant satellites. 

\vspace{-1em}
\subsection{Earth-Curvature-Aware Walker Constellation Model}
\begin{figure}[t]
\centering
\includegraphics[width=0.65\linewidth]{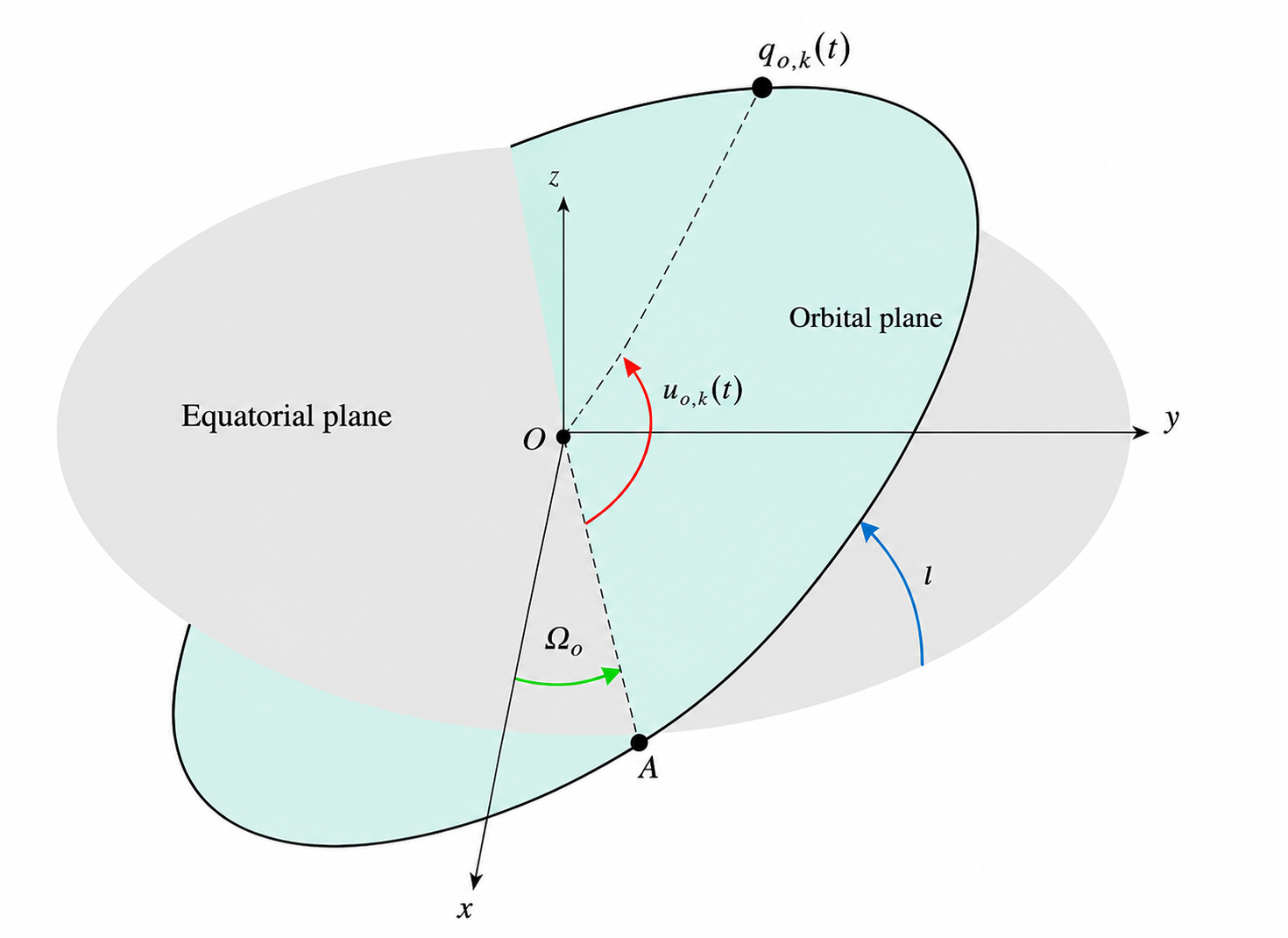} 
\caption{Illustration of longitude and inclination.}
\vspace{-1em}
\end{figure}
To capture practical orbit-constrained LEO geometry, we further incorporate Earth curvature and satellite orbital motion, thereby accounting for the time-varying slant range, line-of-sight (LoS) direction, elevation angle, and satellite visibility. Let $R_E$ denote the Earth radius.  We assume that all satellites follow circular orbits at altitude $H$ and move along their orbits in the same direction with constant angular velocity. Over the considered short observation interval, Earth rotation is neglected, while satellite orbital motion is explicitly retained and the orbital planes are treated as fixed. The Earth-centred radius of the satellite orbit is $R_s = R_E + H$. To describe the Walker constellation geometry, we consider $N_{\mathrm{orb}}$ circular orbital planes, where each orbital plane contains $N_{\mathrm{sat}}$ satellites. The intersection point between an orbital plane and the equatorial plane at which the satellites cross the equator from south to north is referred to as the ascending node of that orbit. The longitude of the $o$-th orbital plane, denoted by $\Omega_o$, is defined as the angle from the reference $x$-axis to the ascending node, measured on the equatorial plane. The common orbital inclination is denoted by $\iota$, which is defined as the angle between the orbital plane and the equatorial plane at the ascending node. Under the fixed-orbit assumption, $\iota$ remains constant over time. For the $k$-th satellite on the $o$-th orbital plane, its orbital phase at time $t$ is denoted by $u_{o,k}(t)$, which is defined as the angle between the satellite and the ascending point at that time, measured along the orbital plane of satellites. The geometric definitions of $\Omega_o$, $\iota$, and $u_{o,k}(t)$ are illustrated in Fig.~2. For circular orbits, the orbital angular velocity is represented by
\vspace{-0.1em}
\begin{equation}
    \omega(H)
    =
    \sqrt{\frac{\mu_E}{(R_E+H)^3}},
    \label{eq:orbital_angular_velocity}
\end{equation}
where $\mu_E$ is the Earth's gravitational parameter.

For a Walker constellation, the orbital longitudes are uniformly spaced as
\begin{equation}
    \Omega_o
    =
    \left(
    \bar{\Omega}
    + \frac{2\pi(o-1)}{N_{\mathrm{orb}}}
    \right)\mathrm{mod}\ 2\pi
    \qquad
    o=1,\ldots,N_{\mathrm{orb}},
    \label{eq:walker_orbit_longitude}
\end{equation}
where $\bar{\Omega}$ is the global longitude offset and $\bar{\Omega}$  is uniformly distributed over $(0,2\pi/N_{\mathrm{orb}})$.

The initial orbital phase of the $k$-th satellite on the $o$-th orbital plane is given by
\vspace{-0.1em}
\begin{equation}
    u_{o,k}(0)
    =
   \bar u
    +
    \frac{2\pi(k-1)}{N_{\mathrm{sat}}}
    +
    \frac{2\pi F(o-1)}{N_{\mathrm{orb}}N_{\mathrm{sat}}},
    \label{eq:walker_initial_phase}
\end{equation}
where $\bar u$ is the global phase offset and $\bar u$  is uniformly distributed over  $(0,2\pi/N_{\mathrm{sat}})$, which is independent of $\bar{\Omega}$. $F\in\{0,1,\ldots,N_{\rm orb}-1\}$ is the Walker phasing factor governing inter-plane satellite phasing~\cite{b33}. 

To capture the time-varying LEO geometry, the orbital phase is modelled as $u_{o,k}(t)=\bigl(u_{o,k}(0)+\omega(H)t\bigr)\bmod 2\pi$. The Earth-centred position of the $k$-th satellite on the $o$-th orbital plane is then expressed as\cite{b18}
\begin{equation}
    \mathbf q_{o,k}(t)
    =
    \mathbf R_z(\Omega_o)\mathbf R_x(\iota)
    \begin{bmatrix}
        R_s\cos u_{o,k}(t)\\
        R_s\sin u_{o,k}(t)\\
        0
    \end{bmatrix},
    \label{eq:walker_satellite_position}
\end{equation}
where $\mathbf R_x(\iota)$ accounts for the orbital inclination, while $\mathbf R_z(\Omega_o)$ accounts for the orbital longitude.

Let $\mathbf q_0$ denote the Earth-centred position of the representative sensing target. When the sensing target has altitude $h_t$,  its Earth-centred radius is $R_0=R_E+h_t$. The Earth-centred central angle between satellite $(o,k)$ and the sensing target is obtained from their 3-D Earth-centred positions as
\begin{equation}
    \vartheta_{o,k}(t_q)
    =
    \arccos
    \left(
    \frac{
    \mathbf q_{o,k}^{T}(t_q)\mathbf q_0
    }
    {R_sR_0}
    \right),
    \label{eq:walker_central_angle}
\end{equation}

Accordingly, the satellite-to-target slant range is
\begin{equation}
    d_{o,k}(H,t_q)
    =
    \sqrt{
    R_s^2+R_0^2
    -
    2R_sR_0\cos\vartheta_{o,k}(t_q)
    } .
    \label{eq:walker_slant_range}
\end{equation}

The corresponding elevation angle satisfies
\begin{equation}
    \sin \epsilon_{o,k}(H,t_q)
    =
    \frac{
    R_s\cos\vartheta_{o,k}(t_q)-R_0
    }
    {d_{o,k}(H,t_q)} .
    \label{eq:walker_elevation_angle}
\end{equation}

The visibility condition $\epsilon_{o,k}(H,t_q)\geq \epsilon_{\min}$ is equivalently expressed as $\vartheta_{o,k}(t_q)\leq\vartheta_{\max}(H,\epsilon_{\min})$. Therefore
\begin{equation}
    \vartheta_{\max}(H,\epsilon_{\min})
    =
    \arccos\left(
    \frac{R_0}{R_s}\cos\epsilon_{\min}
    \right)
    -
    \epsilon_{\min}.
    \label{eq:walker_theta_max}
\end{equation}

Moreover, the visibility indicator is defined as
\begin{equation}
    \chi_{o,k}(t_q)
    =
    \begin{cases}
    1, & \vartheta_{o,k}(t_q)\leq
    \vartheta_{\max}(H,\epsilon_{\min}),\\
    0, & \vartheta_{o,k}(t_q)>
    \vartheta_{\max}(H,\epsilon_{\min}).
    \end{cases}
    \label{eq:walker_visibility_indicator}
\end{equation}

The satellite-assisted bistatic range associated with satellite $(o,k)$ and radar receiver $n$ at time $t_q$ is
\begin{equation}
    \rho^{\mathrm{sat}}_{o,k,n,i}(t_q)
    =
    d_{o,k}(H,t_q)
    +
    \left\|
    \mathbf p_i-\mathbf r_n
    \right\|,
    \label{eq:walker_bistatic_range}
\end{equation}

The range gradient vector can be expressed as
\begin{equation}
    \mathbf g^{\mathrm{sat}}_{o,k,n,i}(t_q)
    =
    \mathbf a_{o,k}(t_q)
    +
    \frac{
    \mathbf p_i-\mathbf r_n
    }
    {
    \left\|
    \mathbf p_i-\mathbf r_n
    \right\|
    },
    \label{eq:walker_range_gradient}
\end{equation}
where \begin{equation}
    \mathbf a_{o,k}(t_q)
    =
    \frac{
    \mathbf U_0^T
    \left(
    \mathbf q_0-\mathbf  q_{o,k}(t_q)
    \right)
    }
    {
    d_{o,k}(H,t_q)
    },
    \label{eq:walker_satellite_los_local}
\end{equation}
is the satellite-target LoS unit vector, and $\mathbf U_0^T$ transforms an Earth-centred vector into the local sensing frame.

 Under the effective SCNR model, the simplified effective satellite-assisted weight is $
    \omega^{\mathrm{sat}}_{o,k,i}(t_q)
    =
    \frac{
    K_{\mathrm{eff}}^{\mathrm{sat}}
    }
    {
    d_{o,k}^2(H,t_q)
    }.$ The satellite-assisted FIM under the Walker constellation is obtained by summing the contributions from all visible satellite-time observations and all radar receivers, which is
\begin{equation}
\begin{aligned}
    \mathbf J^{\mathrm W}_{\mathrm{sat},i}
    &=
    \sum_{q=1}^{Q}
    \sum_{o=1}^{N_{\mathrm{orb}}}
    \sum_{k=1}^{N_{\mathrm{sat}}}
    \chi_{o,k}(t_q)
    \omega^{\mathrm{sat}}_{o,k,i}(t_q)  \\
    &\quad \times
    \sum_{n=1}^{N_r}
    \mathbf g^{\mathrm{sat}}_{o,k,n,i}(t_q)
    \left(
    \mathbf g^{\mathrm{sat}}_{o,k,n,i}(t_q)
    \right)^T .
\end{aligned}
\label{eq:walker_satellite_fim}
\end{equation}

Let $M_{\mathrm{vis}}$ denote the mean number of visible satellite-time observations, averaged over $\bar{\Omega}$ and $\bar u$, where each visible satellite-time pair is counted as one observation. To obtain closed-form moments, the visible satellite-time observations are further approximated as being uniformly distributed over the visible spherical cap above the representative target. For the second-order fluctuation analysis, cross-observation covariances induced by the common orbital structure and Walker phasing are neglected.

\begin{proposition}
Let $  x=\cos\vartheta$ and $x_0=\cos\vartheta_{\max}(H,\epsilon_{\min})$, where $\vartheta$ is the Earth-centred central angle.  Under the uniform visible-cap approximation, $x$ is uniformly distributed over $[x_0,1]$. The corresponding SCNR-weighted moments are
\begin{equation}
    \mu_m
    =
    \frac{K_{\mathrm{eff}}^{\mathrm{sat}}}{1-x_0}
    \int_{x_0}^{1}
    \frac{
    (R_sx-R_0)^m
    }
    {
    D(x)^{m/2+1}
    }
    dx,
    \quad
    m=0,1,2.
    \label{eq:first_order_visible_moments_integral}
\end{equation}
where $D(x)=R_s^2+R_0^2-2R_sR_0x$, $    d^2(x)=D(x)$ and $s(x)= \frac{R_sx-R_0}{\sqrt{D(x)}}$.
    
The averaged satellite-assisted FIM under the isotropic visible-observation approximation can be expressed as
\begin{equation}
    \overline{\mathbf J}^{\mathrm W}_{\mathrm{sat}}
    \approx
    \operatorname{diag}
    \left(
    \bar{\lambda}^{\mathrm{sat}}_{x},
    \bar{\lambda}^{\mathrm{sat}}_{x},
    \bar{\lambda}^{\mathrm{sat}}_{z}
    \right),
    \label{eq:visible_cap_sat_fim_diag}
\end{equation}
where $\bar{\lambda}^{\mathrm{sat}}_{x} = \frac{N_rM_{\mathrm{vis}}}{2} \left[(1+c_r^2)\mu_0-\mu_2\right]$ and $\bar{\lambda}^{\mathrm{sat}}_{z} = N_rM_{\mathrm{vis}} \left[\mu_2-2b_r\mu_1+b_r^2\mu_0\right]$.
\end{proposition}

\begin{IEEEproof}
Please refer to Appendix E.
\end{IEEEproof}

Since the terrestrial FIM under the same symmetric deployment is $ \mathbf J_{\mathrm T}
    =
    \operatorname{diag}
    \left(
    \lambda^{\mathrm{ter}}_{x},
    \lambda^{\mathrm{ter}}_{x},
    \lambda^{\mathrm{ter}}_{z}
    \right)$, the mean hybrid FIM becomes $\overline{\mathbf J}^{\mathrm W}_{\mathrm{hyb}}
    \approx
    \operatorname{diag}
    \left(
    \bar{\lambda}^{\mathrm W}_{x},
    \bar{\lambda}^{\mathrm W}_{x},
    \bar{\lambda}^{\mathrm W}_{z}
    \right)$ with $\bar{\lambda}^{\mathrm W}_{x}= \lambda^{\mathrm{ter}}_{x} + \bar{\lambda}^{\mathrm{sat}}_{x}$  and $
    \bar{\lambda}^{\mathrm W}_{z}
    =
    \lambda^{\mathrm{ter}}_{z}
    +
    \bar{\lambda}^{\mathrm{sat}}_{z}$. The hybrid FIM depends on the global longitude offset $\bar{\Omega}$ and the global phase offset $\bar u$. The corresponding average Walker CRLB is defined as
\begin{equation}
    \overline{\mathrm{CRLB}}^{\mathrm W}
    =
    \mathbb E_{\boldsymbol{\xi}}
    \left[
    \operatorname{tr}
    \left(
    \left(
    \mathbf J^{\mathrm W}_{\mathrm{hyb}}(\boldsymbol{\xi})
    \right)^{-1}
    \right)
    \right],
    \quad
    \boldsymbol{\xi}=(\bar{\Omega},\bar u).
    \label{eq:average_walker_crlb_exact}
\end{equation}
where $ \mathbf J^{\mathrm W}_{\mathrm{hyb}}$ is the hybrid FIM under the Walker constellation.

\begin{proposition}
We define $ \kappa_x
    =
    \frac{1}{\bar{\lambda}^{\mathrm W}_{x}},
    $ and $
   \kappa_z
    =
    \frac{1}{\bar{\lambda}^{\mathrm W}_{z}}$. Under the visible-cap approximation, and when the Walker FIM fluctuation is sufficiently small such that third- and higher-order terms in the inverse-FIM expansion are negligible, the closed-form Taylor approximation of the average Walker CRLB is
\begin{equation}
    \overline{\mathrm{CRLB}}^{\mathrm W}
    \approx
    \frac{2}{\bar{\lambda}^{\mathrm W}_{x}}
    +
    \frac{1}{\bar{\lambda}^{\mathrm W}_{z}}
    +
    \mathcal T^{\mathrm W}.
    \label{eq:visible_closed_form_taylor_crlb}
\end{equation}
where $   
    \mathcal T^{\mathrm W}
    ={}
    2\kappa_x^3
    \mathcal M_{xx}
    +
    \kappa_z^3
    \mathcal M_{zz}
    + 2\kappa_x^3
    \mathcal M_{xy}
    \nonumber
    + 2\left(
        \kappa_x^2\kappa_z
        +
        \kappa_x\kappa_z^2
    \right)
    \mathcal M_{xz}$ and  $\mathcal M_{ij}$ denotes the second-order moment of the corresponding Walker hybrid FIM fluctuation.
\end{proposition}

\begin{IEEEproof}
Please refer to Appendix F.
\end{IEEEproof}

\section{CRLB-Oriented Cooperative Satellite Selection}
This section investigates cooperative satellite selection for the representative target considered in Section III, whose index $i$ is fixed and omitted hereafter. Since localisation performance depends jointly on sensing reliability and geometry, selecting satellites based only on their SCNR may be suboptimal. We first characterise the marginal trace-CRLB reduction of each candidate satellite, analyse the roles of directional complementarity and information saturation, and derive analytical bounds and a sufficient ordering condition. Guided by these results, we propose a CRLB-oriented greedy satellite selection strategy that sequentially minimises the hybrid trace-CRLB. We further compare satellite-only and hybrid-aware selection.
\vspace{-1em}
\subsection{CRLB-Oriented Satellite Selection Criteria}
Let $\mathcal S^{\mathrm W}
=\{1,\ldots,N_{\mathrm{orb}}N_{\mathrm{sat}}\}$ denote the set of all Walker satellites. A satellite is regarded as a selectable candidate if its elevation angle exceeds the minimum visibility threshold during at least one snapshot of the observation window. Accordingly
\begin{equation}
    \mathcal S_{\mathrm{vis}}^{\mathrm W}
    =
    \left\{
        s\in\mathcal S^{\mathrm W}
        \,\middle|\,
        \max_{q\in\{1,\ldots,Q\}}
        \epsilon_s(t_q)
        \geq
        \epsilon_{\min}
    \right\},
    \label{eq:visible_walker_candidate_set}
\end{equation}
where $K_{\mathrm{vis}} = |\mathcal S_{\mathrm{vis}}^{\mathrm W}|$.

For each candidate Walker satellite $s$, $\mathbf J_s^{\mathrm{W,sat}}$ denotes its FIM contribution for the representative target, obtained from (42) by retaining only the contribution of satellite $s$ and summing over all snapshots and radar receivers. For a selected subset $\mathcal S_M\subseteq\mathcal S_{\rm vis}^{\rm W}$ with $|\mathcal S_M|=M$, the corresponding hybrid FIM is
\begin{equation}
    \mathbf J^{\mathrm{hyb}}
    \left(
        \mathcal S_M
    \right)
    =
    \mathbf J^{\mathrm{ter}}
    +
    \sum_{s\in\mathcal S_M}
    \mathbf J_s^{\mathrm{W,sat}},
    \label{eq:selected_subset_hybrid_fim}
\end{equation}
where $\mathbf J^{\mathrm{ter}}$ denotes the terrestrial FIM for the considered receiver deployment.

The corresponding trace-CRLB objective is defined as $
    \mathcal G
    \left(
        \mathcal S_M
    \right)
    =
    \operatorname{tr}
    \left[
        \left(
            \mathbf J^{\mathrm{hyb}}
            \left(
                \mathcal S_M
            \right)
        \right)^{-1}
    \right] $. Accordingly, the optimal satellite subset is given by
\vspace{-0.3em}
\begin{equation}
    \mathcal S_M^{\star}
    =
    \arg\min_{\mathcal S_M\subseteq\mathcal S_{\mathrm{vis}}^{W},\,
    |\mathcal S_M|=M}
    \mathcal G(\mathcal S_M).
    \label{eq:global_crlb_selection_subset}
\end{equation}
\vspace{-0.5em}

Solving~\eqref{eq:global_crlb_selection_subset} by exhaustive search requires checking \(\binom{K_{\mathrm{vis}}}{M}\) possible subsets. This becomes computationally prohibitive when the number of visible candidates is large. Therefore, the exhaustive solution is used as a benchmark, while we adopt a low-complexity greedy procedure.

\subsubsection{{Marginal CRLB Analysis}}
To characterise the localisation improvement provided by each candidate satellite, consider the currently selected subset $\mathcal S_{\mathrm c}\subseteq \mathcal S_{\mathrm{vis}}^{\mathrm W}$, where $\mathcal S_{\mathrm c}$ contains the visible Walker satellites that have already been selected. Consider a remaining candidate satellite $s\in\mathcal S_{\mathrm{vis}}^{\mathrm W} \setminus\mathcal S_{\mathrm c}$. Since $\mathbf J_s^{\mathrm{W,sat}}$ is symmetric positive semidefinite, it admits the factorisation as $\mathbf J_s^{\mathrm{W,sat}} = \mathbf L_s\mathbf L_s^{\mathrm T}$ where $\mathbf L_s=
(\mathbf J_s^{\mathrm{W,sat}})^{1/2}$. After incorporating satellite $s$, the updated hybrid FIM becomes
\begin{equation}
    \mathbf J^{\mathrm{hyb}}
    \left(
        \mathcal S_{\mathrm c}\cup\{s\}
    \right)
    =
    \mathbf J^{\mathrm{hyb}}
    \left(
        \mathcal S_{\mathrm c}
    \right)
    +
    \mathbf L_s\mathbf L_s^{\mathrm T}.
    \label{eq:updated_subset_hybrid_fim}
\end{equation}

Applying the Woodbury matrix identity\cite{b30} gives
\begin{align}
    &
    \left[
        \mathbf J^{\mathrm{hyb}}
        \left(
            \mathcal S_{\mathrm c}\cup\{s\}
        \right)
    \right]^{-1}
    =
    \left[
        \mathbf J^{\mathrm{hyb}}
        \left(
            \mathcal S_{\mathrm c}
        \right)
    \right]^{-1}
    \nonumber\\
    &\quad-
    \left[
        \mathbf J^{\mathrm{hyb}}
        \left(
            \mathcal S_{\mathrm c}
        \right)
    \right]^{-1}
    \mathbf L_s
    \Bigg(
        \mathbf I
        +
        \mathbf L_s^{\mathrm T}
        \left[
            \mathbf J^{\mathrm{hyb}}
            \left(
                \mathcal S_{\mathrm c}
            \right)
        \right]^{-1}
        \mathbf L_s
    \Bigg)^{-1}
    \nonumber\\
    &\qquad\times
    \mathbf L_s^{\mathrm T}
    \left[
        \mathbf J^{\mathrm{hyb}}
        \left(
            \mathcal S_{\mathrm c}
        \right)
    \right]^{-1}.
    \label{eq:woodbury_satellite_selection}
\end{align}

Therefore, the trace-CRLB reduction by adding satellite $s$ is  $\Delta_s
    \left(
        \mathcal S_{\mathrm c}
    \right)
    =
    \mathcal G
    \left(
        \mathcal S_{\mathrm c}
    \right)
    -
    \mathcal G
    \left(
        \mathcal S_{\mathrm c}\cup\{s\}
    \right)
    \nonumber$.  
Substituting \eqref{eq:woodbury_satellite_selection} into the definition of $\Delta_s(\mathcal S_{\mathrm c})$ gives
\vspace{-0.5em}
\begin{align}
\Delta_s(\mathcal S_{\mathrm c})
&=
\operatorname{tr}\!\Bigg[
\left[\mathbf J^{\mathrm{hyb}}(\mathcal S_{\mathrm c})\right]^{-1}
\mathbf L_s
\left(
\mathbf I+\mathbf L_s^{\mathrm T}
\left[\mathbf J^{\mathrm{hyb}}(\mathcal S_{\mathrm c})\right]^{-1}
\mathbf L_s
\right)^{-1}
\nonumber\\
&\qquad\times
\mathbf L_s^{\mathrm T}
\left[\mathbf J^{\mathrm{hyb}}(\mathcal S_{\mathrm c})\right]^{-1}
\Bigg].
\label{eq:incremental_crlb_reduction_psd}
\end{align}
    
\vspace{-0.5em}
The matrix inside the trace in
\eqref{eq:incremental_crlb_reduction_psd} is positive
semidefinite. Hence, $ \Delta_s(\mathcal S_{\mathrm c})\geq 0$. Moreover, since $\mathbf J^{\mathrm{hyb}}(\mathcal S_{\mathrm c})\succ \mathbf 0$, the equality holds if and only if $\mathbf L_s=\mathbf 0$, or equivalently, $\mathbf J_s^{\mathrm{W,sat}}=\mathbf 0$. Therefore, any nonzero candidate-satellite FIM contribution strictly reduces the trace-CRLB.

\subsubsection{{SCNR Ordering and Geometric Complementarity}} We next compare the relative localisation gains of two candidate satellites from the same satellite constellation. We first consider the case where the two satellites provide the same sensing geometry but have different SCNR-dependent information strengths. Consider two candidate satellites $s_1,s_2\in
\mathcal S_{\mathrm{vis}}^{\mathrm W} \setminus\mathcal S_{\mathrm c}$. Suppose that $ \mathbf J_{s_i}^{\mathrm{W,sat}} = \omega_{s_i}^{\mathrm{eff}}  \mathbf J^{\mathrm{geo}}$ for $ i\in\{1,2\}$, where $\mathbf J^{\mathrm{geo}}$ is the common geometry-information matrix and
$\omega_{s_i}^{\mathrm{eff}}$ is the corresponding SCNR-dependent effective information weight.  Therefore, as $\omega_{s_1}^{\mathrm{eff}}
\geq
\omega_{s_2}^{\mathrm{eff}}$, $\mathcal G
    \left(
        \mathcal S_{\mathrm c}\cup\{s_1\}
    \right)
    \leq
    \mathcal G
    \left(
        \mathcal S_{\mathrm c}\cup\{s_2\}
    \right)$. However, when two candidate satellites provide different sensing geometries, satellite $s_1$ is not guaranteed to have a lower trace-CRLB with $\omega_{s_1}^{\mathrm{eff}} > \omega_{s_2}^{\mathrm{eff}}$. This establishes that the selection also depends on how the directional information in $\mathbf J_{s_i}^{\mathrm{geo}}$ complements the information already contained in the current hybrid FIM. This dependence can be characterised exactly through the eigenstructure of the current hybrid FIM. Since $\mathbf J^{\mathrm{hyb}}(\mathcal S_{\mathrm c})$ is symmetric and positive definite, let $\mathbf J^{\mathrm{hyb}}(\mathcal S_{\mathrm c})
=\mathbf V_{\mathrm c}\mathbf E_{\mathrm c}\mathbf V_{\mathrm c}^{\mathrm T}$ with  $\mathbf E_{\mathrm c}
=\operatorname{diag}(e_{\mathrm c,1},e_{\mathrm c,2},e_{\mathrm c,3})$, where the columns of $\mathbf V_{\mathrm c}$ are the orthonormal eigenvectors of the current hybrid FIM and $e_{\mathrm c,j}$ denotes the information strength along the $j$th principal direction.

Define the FIM of the candidate satellite $s$ represented in the eigenbasis of the current hybrid FIM as $ \mathbf Z_s = \mathbf V_{\mathrm c}^{\mathrm T}  \mathbf J_s^{\mathrm{W,sat}} \mathbf V_{\mathrm c}$. Then, the updated FIM after adding satellite $s$ is
\begin{equation}
\mathbf J^{\mathrm{hyb}}(\mathcal S_{\mathrm c})
+\mathbf J_s^{\mathrm{W,sat}}
=
\mathbf V_{\mathrm c}
(\mathbf E_{\mathrm c}+\mathbf Z_s)
\mathbf V_{\mathrm c}^{\mathrm T}.
\label{eq:updated_fim_current_eigenbasis}
\end{equation}

Therefore, marginal trace-CRLB reduction can be expressed exactly as $
    \Delta_s
    \left(
        \mathcal S_{\mathrm c}
    \right)
    =
    \operatorname{tr}
    \left[
        \mathbf E_{\mathrm c}^{-1}
        -
        \left(
            \mathbf E_{\mathrm c}
            +
            \mathbf Z_s
        \right)^{-1}
    \right]
    \label{eq:exact_candidate_gain_eigenbasis}
$, which shows that the marginal trace-CRLB reduction depends jointly on the current directional information strengths and the directional structure of the candidate satellite FIM. Consequently, candidate satellites with comparable overall information strengths can yield different trace-CRLB reductions.

To obtain an interpretable normalised representation, we define $ \mathbf Q_s =\mathbf E_{\mathrm c}^{-1/2}\mathbf Z_s \mathbf E_{\mathrm c}^{-1/2}$, where $\mathbf Q_s$ represents the candidate-satellite information $\mathbf Z_s$ normalised relative to the information already contained in $\mathbf E_{\mathrm c}$. Then, $\left( \mathbf E_{\mathrm c}+\mathbf Z_s \right)^{-1} =\mathbf E_{\mathrm c}^{-1/2} \left(\mathbf I+\mathbf Q_s \right)^{-1} \mathbf E_{\mathrm c}^{-1/2}$ can be expressed. Therefore, the trace-CRLB reduction can be shown as
\begin{align}
    \Delta_s
    \left(
        \mathcal S_{\mathrm c}
    \right)
    &=
    \operatorname{tr}
    \Big[
        \mathbf E_{\mathrm c}^{-1/2}
        \mathbf Q_s
        (\mathbf I+\mathbf Q_s)^{-1}
        \mathbf E_{\mathrm c}^{-1/2}
    \Big].
    \label{eq:normalised_exact_marginal_gain}
\end{align}

To reveal how the normalised candidate information contributes to the marginal trace-CRLB reduction, let $\mathbf Q_s=\mathbf U_s\mathbf D_s\mathbf U_s^{\mathrm T}$ with $\mathbf D_s=\operatorname{diag}(q_{s,1},q_{s,2},q_{s,3})$, where $\mathbf U_s$ contains the orthonormal eigenvectors of $\mathbf Q_s$, and $q_{s,i}$ denotes the corresponding eigenvalue.  It then follows that
\begin{equation}
\mathbf Q_s(\mathbf I+\mathbf Q_s)^{-1}
=
\mathbf U_s
\operatorname{diag}\!\left(
\frac{q_{s,1}}{1+q_{s,1}},
\frac{q_{s,2}}{1+q_{s,2}},
\frac{q_{s,3}}{1+q_{s,3}}
\right)
\mathbf U_s^{\mathrm T}.
\label{eq:effective_normalised_candidate_information}
\end{equation}

Substituting
\eqref{eq:effective_normalised_candidate_information}
into \eqref{eq:normalised_exact_marginal_gain} and applying
the cyclic property of the trace gives
\vspace{-1em}
\begin{align}
    \Delta_s(\mathcal S_{\mathrm c})
   =
    \sum_{i=1}^{3}
    \frac{q_{s,i}}{1+q_{s,i}}
    \sum_{j=1}^{3}
    \frac{
        [\mathbf U_s]_{ji}^{2}
    }{
        e_{\mathrm c,j}
    }.
    \label{eq:exact_directional_incremental_gain}
\end{align}

\vspace{-0.5em}
Consequently, for candidate information with the same effective strength, allocating a larger fraction towards a direction with a smaller $e_{\mathrm c,j}$ yields a larger contribution to the trace-CRLB reduction, which shows the importance of considering whether the new information is directionally complementary to the current hybrid sensing configuration. Therefore, the exact selection metric inherently favours candidate information that supplements weakly informed directions while accounting for the saturation of increasingly strong information modes.

\subsubsection{{Bounds and Ordering Conditions for Marginal CRLB Reduction}}
Building on the exact marginal trace-CRLB reduction, we next establish
analytical bounds on $\Delta_s(\mathcal S_{\mathrm c})$ and use them to derive a sufficient condition for ordering candidate satellites. Since $\mathbf Q_s\succeq\mathbf 0$, its eigenvalues satisfy $\lambda_{\min}(\mathbf Q_s)\le q_{s,i}\le \lambda_{\max}(\mathbf Q_s)$ for $i\in\{1,2,3\}$, where $\lambda_{\min}(\mathbf Q_s)=\min_i q_{s,i}$ and $\lambda_{\max}(\mathbf Q_s)=\max_i q_{s,i}$. Since $q_{s,i}\geq0$, we have $\frac{q_{s,i}}{1+\lambda_{\max}(\mathbf Q_s)}
\leq \frac{q_{s,i}}{1+q_{s,i}}
\leq \frac{q_{s,i}}{1+\lambda_{\min}(\mathbf Q_s)}$. Hence
\begin{equation}
    \frac{
        \mathbf Q_s
    }{
        1+\lambda_{\max}(\mathbf Q_s)
    }
    \preceq
    \mathbf Q_s
    \left(
        \mathbf I+\mathbf Q_s
    \right)^{-1}
    \preceq
    \frac{
        \mathbf Q_s
    }{
        1+\lambda_{\min}(\mathbf Q_s)
    }.
    \label{eq:matrix_bound}
\end{equation}

Applying the congruence transformation $\mathbf E_{\mathrm c}^{-1/2}(\cdot)
\mathbf E_{\mathrm c}^{-1/2}$ to \eqref{eq:matrix_bound} and taking the trace preserves the ordering, which gives
\begin{equation}
    \frac{
        \Psi_s(\mathcal S_{\mathrm c})
    }{
        1+\lambda_{\max}(\mathbf Q_s)
    }
    \leq
    \Delta_s(\mathcal S_{\mathrm c})
    \leq
    \frac{
        \Psi_s(\mathcal S_{\mathrm c})
    }{
        1+\lambda_{\min}(\mathbf Q_s)
    }.
    \label{eq:marginal_gain_bound}
\end{equation}
where $ \Psi_s(\mathcal S_{\mathrm c}) = \operatorname{tr}\left[ \mathbf E_{\mathrm c}^{-1/2} \mathbf Q_s\mathbf E_{\mathrm c}^{-1/2} \right]$.

The preceding bounds further provide a sufficient condition for strictly ordering two candidate satellites. For two candidates $s_1$ and $s_2$, if
\begin{equation}
    \frac{
        \Psi_{s_1}(\mathcal S_{\mathrm c})
    }{
        1+\lambda_{\max}(\mathbf Q_{s_1})
    }
    >
    \frac{
        \Psi_{s_2}(\mathcal S_{\mathrm c})
    }{
        1+\lambda_{\min}(\mathbf Q_{s_2})
    },
    \label{eq:sufficient_ordering_condition}
\end{equation}
then $ \Delta_{s_1}(\mathcal S_{\mathrm c}) > \Delta_{s_2}(\mathcal S_{\mathrm c})$.

\subsubsection{Greedy Satellite Selection}
Guided by the preceding analysis, we construct a greedy selection procedure that, at each iteration, selects the candidate providing the largest exact marginal trace-CRLB reduction. Starting from $\widehat{\mathcal S}_0=\emptyset$, the proposed CRLB-greedy method selects, at the \(m\)th step, the satellite that yields the smallest trace-CRLB after being added to the currently selected set
\begin{equation}
    s_m^{\mathrm{CRLB}}
    =
    \arg\min_{s\in
    \mathcal S_{\mathrm{vis}}^{W}
    \setminus
    \widehat{\mathcal S}_{m-1}}
    \mathcal G
    \left(
    \widehat{\mathcal S}_{m-1}\cup\{s\}
    \right),
    \quad m=1,\ldots,M .
    \label{eq:crlb_greedy_step}
\end{equation}

The selected subset is updated as  $\widehat{\mathcal S}_m
    =
    \widehat{\mathcal S}_{m-1}
    \cup
    \left\{
    s_m^{\mathrm{CRLB}}
    \right\}$. After \(M\) iterations, \(\widehat{\mathcal S}_M\) is returned as the proposed CRLB-oriented cooperative satellite subset. Unlike single-factor selection rules, this criterion evaluates the SCNR-weighted hybrid FIM and therefore jointly accounts for sensing-link reliability and geometry-dependent information diversity.

\subsubsection{Benchmark Selection Criteria}
For comparison, two benchmark satellite-selection criteria are also considered. The first benchmark is the strongest-SCNR criterion, where for each visible candidate satellite \(s\), we define
\vspace{-0.5em}
\begin{equation}
    \eta_s
    =
    \sum_{q=1}^{Q}
    \chi_s(t_q)
    \sum_{n=1}^{N_r}
    \mathrm{SCNR}^{\mathrm{sat}}_{s,n}(t_q).
    \label{eq:aggregate_scnr_strength}
\end{equation}
\vspace{-0.8em}

The strongest-SCNR method selects the \(M\) satellites with the largest values of \(\eta_s\). This criterion favours satellites with strong sensing links, but it does not explicitly account for geometric diversity. The second benchmark is the angular-diversity criterion, which evaluates angular separation among the selected satellites. Let $\phi_s$ denote the representative target-to-satellite azimuth of satellite $s$, obtained by circularly averaging its azimuths over the visible snapshots. Since angular diversity cannot be evaluated before any satellite has been selected, the procedure is initialised using the satellite with the largest aggregate SCNR as  $s_1^{\mathrm{ang}} = \arg\max_{s\in\mathcal S_{\mathrm{vis}}^{W}} \eta_s$ and $
    \widehat{\mathcal S}_1^{\mathrm{ang}}
    =
    \left\{
        s_1^{\mathrm{ang}}
    \right\}$. The remaining satellites are selected sequentially according to the max-min angular separation rule
\begin{equation}
    s_m^{\mathrm{ang}}
    =
    \arg\max_{s\in
    \mathcal S_{\mathrm{vis}}^{W}
    \setminus
    \widehat{\mathcal S}_{m-1}^{\mathrm{ang}}}
    \min_{\ell\in
    \widehat{\mathcal S}_{m-1}^{\mathrm{ang}}}
    \Delta_{\phi}(s,\ell),
    \qquad m=2,\ldots,M,
    \label{eq:angular_greedy_selection}
\end{equation}
where $ \Delta_{\phi}(s,\ell) =\min \left\{|\phi_s-\phi_\ell|, 2\pi-|\phi_s-\phi_\ell|\right\}$, which denotes the circular angular separation between satellites $s$ and $\ell$. The selected subset is updated as $\widehat{\mathcal S}_m^{\mathrm{ang}} =
    \widehat{\mathcal S}_{m-1}^{\mathrm{ang}}
    \cup
    \left\{
        s_m^{\mathrm{ang}}
    \right\}$.
\vspace{-1em}
\subsection{Satellite-Only and Hybrid-Aware Satellite Selection} 
To investigate the role of terrestrial-satellite information complementarity, we compare satellite-only and hybrid-aware selection. For each Walker realisation, both methods operate on the same visible satellite candidate set and use the same satellite-assisted FIMs. Hence, both selection rules use the CRLB-oriented method developed in Section IV.A. Their distinction is whether the existing terrestrial FIM is incorporated during satellite selection. The hybrid-aware strategy includes the existing terrestrial FIM when evaluating each candidate satellite and therefore favours satellites whose information complements the terrestrial sensing geometry. Symmetric radar receiver deployment and asymmetric radar receiver deployment are considered. Under the symmetric deployment, the terrestrial sensing geometry exhibits a relatively weak horizontal directional preference. The symmetric deployment provides relatively balanced horizontal terrestrial information, while the asymmetric deployment produces a stronger directional imbalance. This controlled comparison allows the effect of terrestrial information complementarity on satellite selection to be examined.

\vspace{-0.7em}
\section{Numerical Results}
\label{sec:numerical_results}
\subsection{Simulation Setup}
The carrier frequency is set to \(f_c=15\) GHz. The receiver noise is characterised by a noise figure of \(\mathrm{NF}=8\) dB and a noise spectral density of \(N_0=-174\) dBm/Hz. The TBS transmit power is \(P_{\mathrm{TBS}}=30\) dBm. The LEO satellites operate at an altitude of \(H=500\) km. The transmit power of each satellite is \(P_{\mathrm{sat}}=55\) dBm. The antenna numbers are set to \(N_{\mathrm{TX}}=30\), \(N_{\mathrm{RX}}=4\), and \(N_{\mathrm{SAT}}=8\). The terrestrial ISAC system serves \(N_{\mathrm{com}}=2\) communication users and senses \(N_{\mathrm{tar}}=4\) UAV targets. The target height and TBS height are set to $h_t=80$ m and $h_T=10$ m, respectively. The Walker constellation adopts an orbital inclination of $\iota=53^\circ$. The sensing interval is set to $T_{\mathrm{obs}}=60$ s and is uniformly sampled at $Q=5$ observation instants. A minimum satellite elevation angle of $\epsilon_{\min}=10^\circ$ is imposed for satellite visibility.

\begin{figure}
\centering
\includegraphics[width=0.8\linewidth]{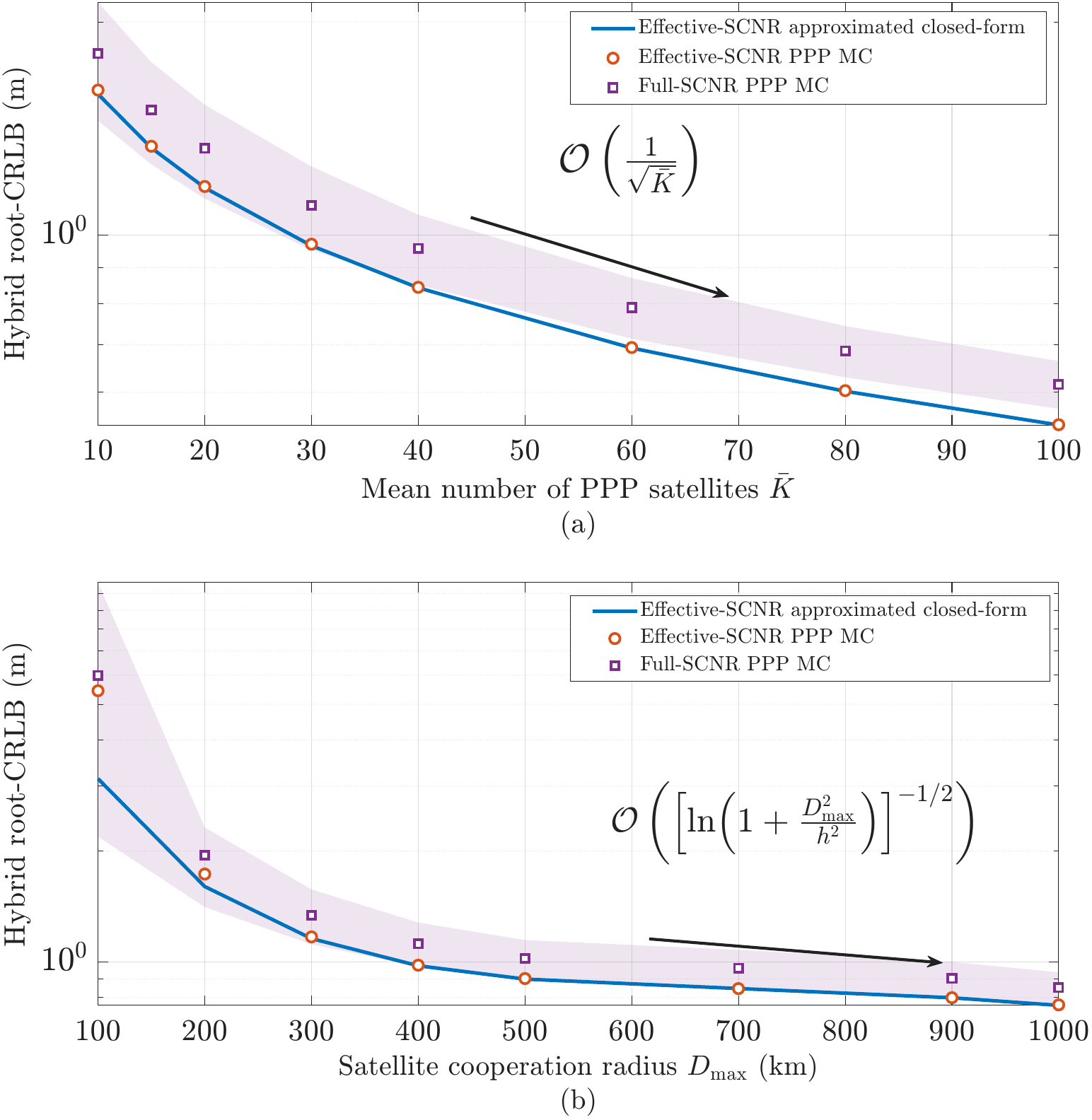} 
\vspace{-1.1em}\caption{Validation and scaling laws of the PPP-averaged hybrid root-CRLB versus (a) the mean number of cooperative satellites $\bar K$ and (b) the satellite cooperation radius $D_{\max}$. The shaded regions show the 10th-90th percentile of the full-SCNR PPP Monte Carlo results.}
\vspace{-1em}
\end{figure}

\begin{figure}[t]
\centering
\includegraphics[width=0.95\linewidth]{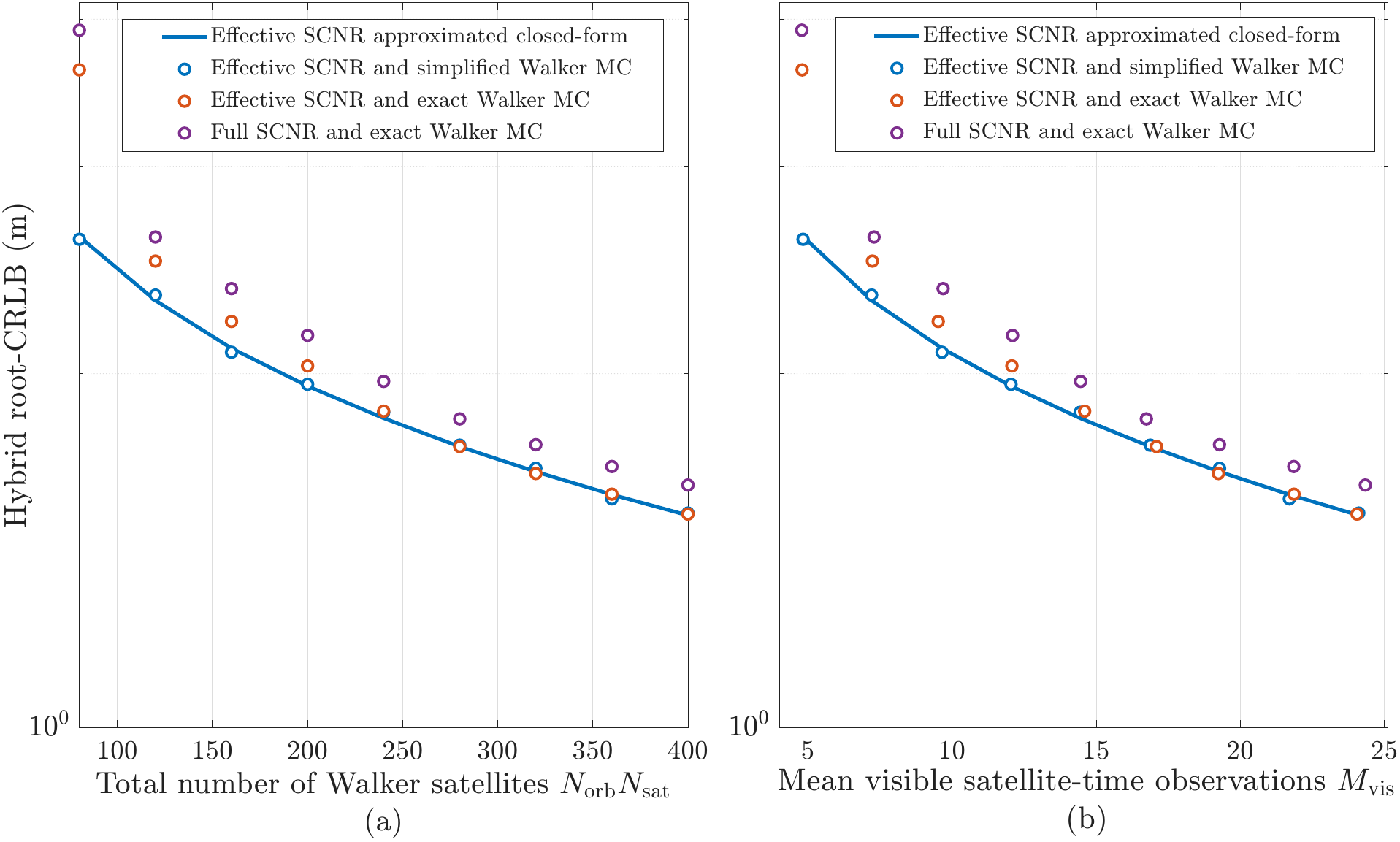} 
\vspace{-1em}\caption{Walker-based hybrid root-CRLB versus (a) the number of Walker satellites $N_{\mathrm{orb}}N_{\mathrm{sat}}$ and (b) the mean number of visible satellite-time observations $M_{\mathrm{vis}}$, comparing the effective-SCNR closed-form approximation with the simplified-Walker, exact-Walker, and full-SCNR Monte Carlo results.}
\vspace{-1.5em}
\end{figure}

\begin{figure*}
\centering
\includegraphics[width=0.8\linewidth]{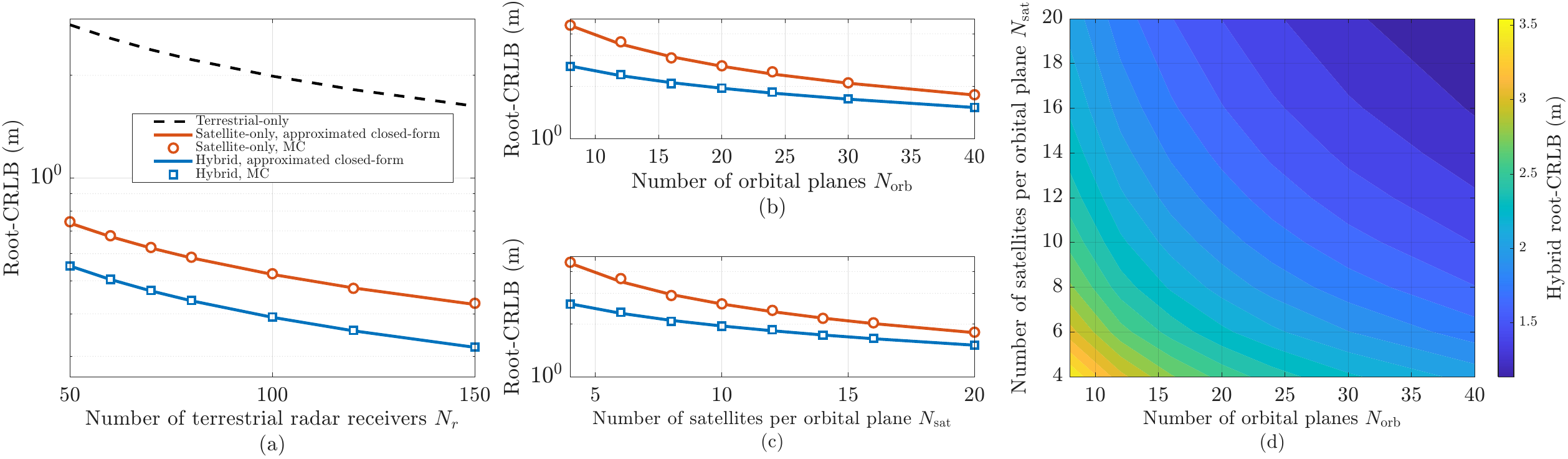} 
\vspace{-1em}\caption{Root-CRLB performance with respect to (a) the number of terrestrial radar receivers $N_r$, (b) the number of orbital planes $N_{\mathrm{orb}}$, and (c) the number of satellites per orbital plane $N_{\mathrm{sat}}$. (d) Hybrid root-CRLB over the two-dimensional Walker constellation design space spanned by $N_{\mathrm{orb}}$ and $N_{\mathrm{sat}}$.}
\vspace{-1em}
\end{figure*}
\vspace{-1em}
\subsection{Sensing Performance}
The following results assess the accuracy of the analytical approximations against Monte Carlo simulations under both the PPP and Walker models.

\subsubsection{Analytical Validation}
Fig.~3(a) shows that the full-SCNR Monte Carlo result closely follows the decay trend predicted by the effective-SCNR approximation, with the hybrid root-CRLB decreasing as $\mathcal{O}\!\left(1/\sqrt{\bar{K}}\right)$ as $\bar{K}$ increases. The contraction of the shaded percentile region with increasing $\bar{K}$ further indicates reduced sensitivity to individual PPP realisations as more cooperative satellites contribute to sensing. Fig.~3(b) further illustrates the effect of the cooperation radius $D_{\max}$. For $D_{\max}\ll h$, the rapid initial reduction in the root-CRLB results from the rapid accumulation of satellite-assisted Fisher information as the cooperation region expands. As $D_{\max}$ increases further, the root-CRLB approaches the asymptotic scaling $\mathcal{O}\!\left(\left[\ln\!\left(1+D_{\max}^{2}/h^{2}\right)\right]^{-1/2}\right)$, revealing progressively weaker marginal localisation gains from incorporating increasingly distant satellites. Fig.~4 evaluates the Walker-based approximation under practical orbital geometry and visibility constraints. As shown in Fig.~4(a), the hybrid root-CRLB decreases as the Walker constellation becomes denser under all considered models. Incorporating the exact Walker geometry introduces a moderate deviation, particularly in the sparse-constellation regime, while the full-SCNR model leads to further deviation without changing the overall decreasing trend. Fig.~4(b) further shows that the root-CRLB decreases with the mean number of visible satellite-time observations $M_{\mathrm{vis}}$ under all considered models, indicating that the sensing gain is directly associated with the number of satellite observations that effectively contribute to the FIM.
\begin{figure}[!t]
\centering
\vspace{-0.5em}
\includegraphics[width=0.8\linewidth]{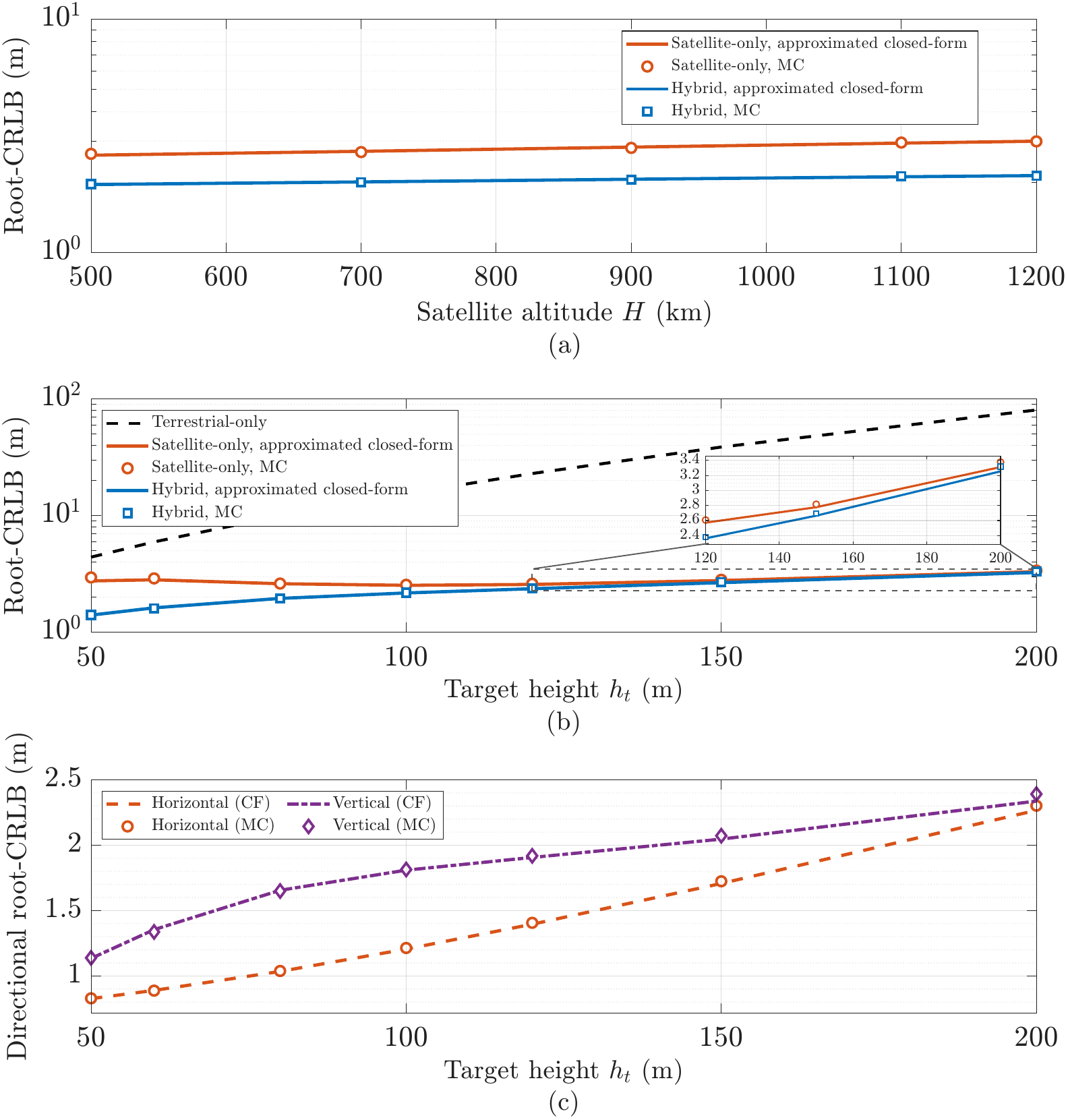} 
\vspace{-1em}\caption{Root-CRLB performance versus (a) satellite altitude $H$, (b) target height $h_t$, and (c) horizontal and vertical root-CRLB versus $h_t$.}
\vspace{-0.8em}
\end{figure}

\begin{figure}[!t]
\centering
\includegraphics[width=0.78\linewidth]{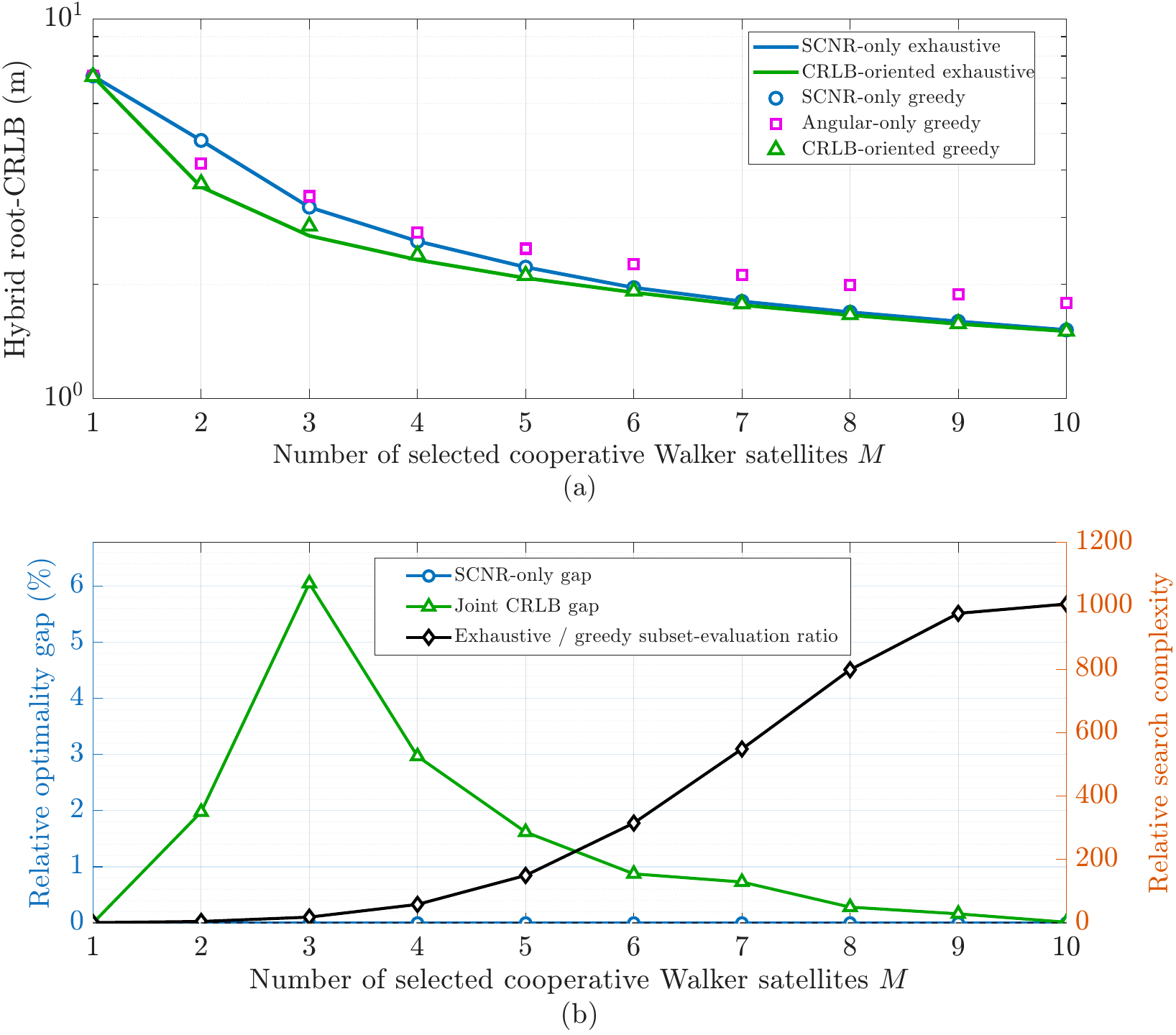} 
\vspace{-1em}\caption{Performance and complexity of cooperative Walker satellite selection. 
(a) Hybrid 3-D root-CRLB versus the number of selected satellites $M$. 
(b) Greedy optimality gap and relative search complexity versus exhaustive search.}
\vspace{-1em}
\end{figure}

\begin{figure}
\centering
\includegraphics[width=0.65\linewidth]{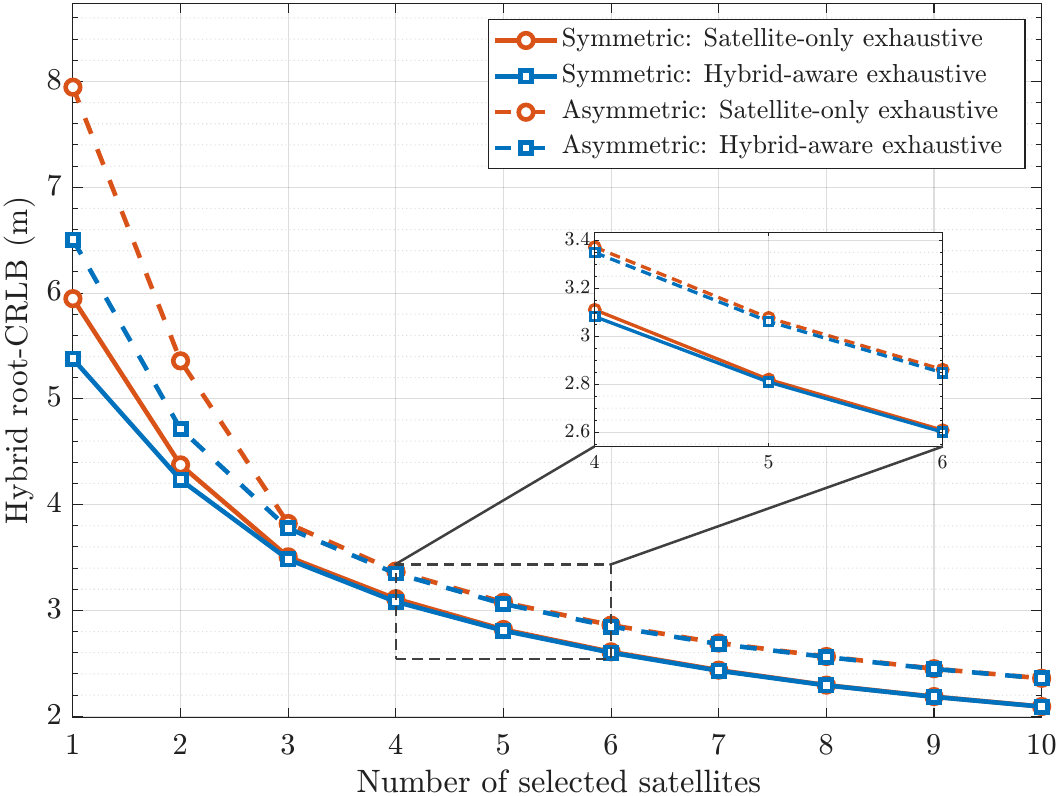} 
\vspace{-1em}\caption{Hybrid root-CRLB versus the number of selected satellites $M$ under symmetric and asymmetric terrestrial receiver deployments for satellite-only and hybrid-aware exhaustive selection.}
\vspace{-1em}
\end{figure}

\subsubsection{Geometric Effects}
Fig.~5(a) shows that increasing the number of terrestrial radar receivers $N_r$ reduces the root-CRLB. The hybrid sensing scheme consistently achieves the lowest root-CRLB, which demonstrates the benefit of combining terrestrial and satellite sensing information. Figs.~5(b) and 5(c) further show that increasing $N_{\mathrm{orb}}$ or $N_{\mathrm{sat}}$ improves the hybrid localisation accuracy by providing richer satellite observations and spatial diversity. However, the gradual flattening of both curves indicates diminishing localisation gains as the Walker constellation becomes denser.  Fig.~5(d) provides a joint view of this behaviour over the $(N_{\mathrm{orb}},N_{\mathrm{sat}})$ space, where the pronounced variation in the sparse-constellation region indicates greater sensitivity to constellation expansion, while the smoother variation in the dense region reflects progressively weaker marginal localisation gains. Following LEO deployments in \cite{b31,b32}, we set $H=500$ km as the lower bound of the considered satellite-altitude range. A higher orbit enlarges visibility but increases satellite-target distance and weakens the sensing link. Fig.~6(a) shows that the root-CRLB generally increases with $H$, indicating that the propagation-loss effect is dominant in the considered setting. The hybrid scheme consistently outperforms the satellite-only scheme by exploiting the additional terrestrial sensing information. Fig.~6(b) further shows that the terrestrial-only root-CRLB degrades rapidly with the target height $h_t$, as the predominantly ground-based receiver geometry provides increasingly limited vertical information. In contrast, the satellite-assisted branch is less sensitive to $h_t$ due to its high-elevation sensing directions, thereby enabling the hybrid scheme to substantially mitigate this degradation. This behaviour is further confirmed by the directional decomposition in Fig.~6(c), where the vertical root-CRLB remains higher than the horizontal component, identifying the vertical direction as the weaker localisation dimension.

\subsubsection{Satellite Selection}
Fig.~7(a) compares the SCNR-only, angular-only, and CRLB-oriented selection criteria. The CRLB-oriented criterion achieves the lowest root-CRLB, particularly when only a few satellites are selected, because it jointly accounts for SCNR-weighted information strength and directional complementarity with the existing hybrid FIM. Fig.~7(b) further evaluates the performance and search complexity of the greedy selection. The SCNR-only criterion exhibits zero optimality gap owing to its additive metric, whereas the CRLB-oriented gap initially increases because the greedy procedure retains earlier selections while exhaustive search independently re-optimises the subset for each $M$. As more satellites are selected, the additional sensing directions reduce the influence of these early decisions and the gap decreases towards zero. Meanwhile, exhaustive search requires up to approximately three orders of magnitude more subset evaluations than the greedy procedure. These results demonstrate that the proposed CRLB-oriented greedy selection achieves near-exhaustive localisation performance with substantially lower search complexity. Fig.~8 compares satellite-only and hybrid-aware exhaustive selection under symmetric and asymmetric terrestrial receiver deployments, where exhaustive search is adopted to isolate the effect of terrestrial-information awareness. The hybrid-aware criterion consistently achieves a lower root-CRLB by incorporating the existing terrestrial FIM when evaluating candidate satellites. As $M$ increases, the performance gap narrows because the larger satellite subset naturally provides richer spatial diversity. The asymmetric deployment generally yields a higher root-CRLB due to its more directionally imbalanced terrestrial sensing geometry. More importantly, the larger gap between satellite-only and hybrid-aware selection under this deployment indicates that terrestrial-information awareness becomes more beneficial when weakly informed directions are more pronounced, as hybrid-aware selection can preferentially choose satellites that provide complementary information along these directions. These results highlight the importance of joint terrestrial-satellite geometry in satellite selection.

\vspace{-0.8em}
\section{Conclusion}
This paper developed a cooperative LEO-terrestrial multistatic ISAC framework for 3-D target localisation. A PPP-based stochastic model enabled tractable hybrid CRLB approximations and scaling laws, revealing fundamentally different localisation gains from increasing nearby satellite density and enlarging the cooperation region. An Earth-curvature-aware Walker model was further developed to capture orbit-constrained, time-varying sensing geometry. Based on the resulting FIM structure, a CRLB-oriented greedy satellite-selection strategy was proposed to jointly exploit sensing reliability and directional information complementarity. Numerical results validated the analytical approximations and demonstrated near-exhaustive selection performance at substantially lower search complexity, with hybrid-aware selection providing greater benefits under directionally imbalanced terrestrial deployments.

\vspace{-1.5em}
\appendices
\section{Proof of Proposition 1: }
Using the second-order Campbell theorem for the PPP, the second-order term in the Taylor approximation can be expressed as
\begin{equation}
\operatorname{tr}\!\left(
\bar{\mathbf J}^{-1}\Delta\mathbf J
\bar{\mathbf J}^{-1}\Delta\mathbf J
\bar{\mathbf J}^{-1}
\right)
=
\sum_{i\in\{x,y,z\}}
\sum_{j\in\{x,y,z\}}
A_i^2A_j(\Delta J_{ij})^2,
\end{equation}
where $\bar{\mathbf J}^{-1}=\operatorname{diag}(A_x,A_y,A_z)$.

Therefore, the second moments can be obtained using $Q_{ij}
    = \mathbb E_{\rho,\phi_s} \left[J_{s,ij}^2(\rho,\phi_s)\right]$ as
\begin{align}
    Q_{xx}
    &=
    N_r^2
    \left(
    \frac{3}{8}\xi_{\omega aa}
    +
    \frac{c_r^2}{2}\xi_{\omega a}
    +
    \frac{c_r^4}{4}\xi_{\omega}
    \right), 
    \\
    Q_{xy} 
   & =
    \frac{N_r^2}{8}
    \xi_{\omega aa},\quad
    Q_{xz}
    =
    \frac{N_r^2}{2}
    \xi_{\omega a\zeta},  \quad  Q_{zz}=
    N_r^2
    \xi_{\omega\zeta},
\end{align}

By the same azimuthal-averaging argument, $Q_{yy}=Q_{xx}$ and $Q_{yz}=Q_{xz}$ where the second-order SCNR-weighted radial moments in the second moments are defined as
\begin{align}
    \xi_{\omega}
    &=
    \mathbb E_{\rho}
    \left[
    \omega^2(\rho)
    \right], \quad
    \xi_{\omega a}
    =
    \mathbb E_{\rho}
    \left[
    \omega^2(\rho)a_s^2(\rho)
    \right],\\
    \xi_{\omega aa}
    &=
    \mathbb E_{\rho}
    \left[
    \omega^2(\rho)a_s^4(\rho)
    \right],\quad
    \xi_{\omega \zeta}
    =
    \mathbb E_{\rho}
    \left[
    \omega^2(\rho)\zeta_s^4(\rho)
    \right],\\
    \xi_{\omega a\zeta}
    &=
    \mathbb E_{\rho}
    \left[
    \omega^2(\rho)a_s^2(\rho)\zeta_s^2(\rho)
    \right],
\end{align}

By using the function $ \mathcal I(\nu)$, the second-order radial moments can be expressed as
\begin{align}
    \xi_{\omega}
    &=
    \left(K_{\rm eff}^{\rm sat}\right)^2
    \mathcal I(-2),\\
    \xi_{\omega a}
    &=
    \left(K_{\rm eff}^{\rm sat}\right)^2
    \left[
    \mathcal I(-2)
    -
    h^2\mathcal I(-3)
    \right],\\
    \xi_{\omega aa}
    &=
    \left(K_{\rm eff}^{\rm sat}\right)^2
    \left[
    \mathcal I(-2)
    -
    2h^2\mathcal I(-3)
    +
    h^4\mathcal I(-4)
    \right],\\
    \xi_{\omega a\zeta}
    &=
    \left(K_{\rm eff}^{\rm sat}\right)^2
    \Big[
    b_r^2
    \left(
    \mathcal I(-2)-h^2\mathcal I(-3)
    \right)
    \nonumber\\
    &\quad
    -
    2b_rh
    \left(
    \mathcal I(-5/2)-h^2\mathcal I(-7/2)
    \right)
    \nonumber\\
    &\quad
    +
    h^2
    \left(
    \mathcal I(-3)-h^2\mathcal I(-4)
    \right)
    \Big],\\
    \xi_{\omega \zeta}
    &=
    \left(K_{\rm eff}^{\rm sat}\right)^2
    \Big[
    b_r^4\mathcal I(-2)
    -
    4b_r^3h\mathcal I(-5/2)
    \nonumber\\
    &\quad
    +
    6b_r^2h^2\mathcal I(-3)
    -
    4b_rh^3\mathcal I(-7/2)
    +
    h^4\mathcal I(-4)
    \Big].
\end{align}

Substituting these second-order moments into $Q_{ij}$ completes the proof.
\vspace{-0.8em}
\section{Proof of Theorem 1}

For fixed \(D_{\max}\) and \(H\), the quantities \(\mu_x\) and \(\mu_z\) are independent of \(\bar K\). For the $\bar K$-scaling analysis, we consider $\bar K\to\infty$ while keeping $D_{\max}$ and $H$ fixed. The first term asymptotically becomes $A_x + A_y + A_z
    \sim
    \frac{1}{\bar K}
    \left(
        \frac{2}{\mu_x}
        +
        \frac{1}{\mu_z}
    \right)$, since $ A_x = \mathcal O(\bar K^{-1}), A_y = \mathcal O(\bar K^{-1}),$ and $ A_z = \mathcal O(\bar K^{-1})$, each term in $\mathcal T_{\mathrm{taylor}}$ contains a product of three inverse-FIM diagonal entries. Therefore, the correction term can be expressed as $\mathcal T_{\mathrm{taylor}} =\mathcal O(\bar K^{-3})$, and $  \bar K \mathcal T_{\mathrm{taylor}}
 = \mathcal O(\bar K^{-2}),$ which decays faster than the first-order CRLB. Therefore, the root-CRLB decreases proportionally to  $1/\sqrt{ {\bar{K}}}$, which completes the proof.

\vspace{-0.8em}
\section{Proof of Lemma 1}
Using $  \delta_D = \frac{D_{\max}^2}{h^2}$, the horizontal satellite-assisted FIM contribution in (26) can be rewritten as
\begin{equation}
\begin{aligned}
&\bar K(D_{\max})\mu_x(D_{\max}) \\
&\quad =
\frac{
    \lambda_s\pi N_{\mathrm r}
    K_{\mathrm{eff}}^{\mathrm{sat}}
}{2}
\left[
    (1+c_r^2)\ln(1+\delta_D)
    -
    \frac{\delta_D}{1+\delta_D}
\right],
\end{aligned}
\end{equation}

For $D_{\max}\ll h$, we have
\begin{equation}
\begin{aligned}
&\bar K(D_{\max})\mu_x(D_{\max}) \\
&\quad =
\frac{
    \lambda_s\pi N_{\mathrm r}
    K_{\mathrm{eff}}^{\mathrm{sat}}
}{2}
\left[
    c_r^2\delta_D
    +
    \frac{1-c_r^2}{2}\delta_D^2
    +
    \mathcal O(\delta_D^3)
\right],
\end{aligned}
\end{equation}

Hence, $ \bar K(D_{\max})\mu_x(D_{\max}) = \lambda_s\pi N_{\mathrm r}
K_{\mathrm{eff}}^{\mathrm{sat}}\mathcal O\left( \frac{D_{\max}^2}{h^2}
\right)$. 

For $D_{\max}\gg h$,  $ \frac{\delta_D}{1+\delta_D} =  1+\mathcal O(\delta_D^{-1})$, which gives
\begin{equation}
\begin{aligned}
&\bar K(D_{\max})\mu_x(D_{\max}) \\
&\quad =
\frac{
    \lambda_s\pi N_{\mathrm r}
    K_{\mathrm{eff}}^{\mathrm{sat}}
}{2}
\left[
    (1+c_r^2)\ln(1+\delta_D)
    -
    1
    +
    \mathcal O(\delta_D^{-1})
\right],
\end{aligned}
\end{equation}

The logarithmic term dominates as $\delta_D\rightarrow\infty$,
which yields
\begin{equation}
\bar K(D_{\max})\mu_x(D_{\max})
=
\lambda_s\pi N_{\mathrm r}
K_{\mathrm{eff}}^{\mathrm{sat}}
\mathcal O\left[
    \ln\left(
        1+\frac{D_{\max}^2}{h^2}
    \right)
\right].
\end{equation}

These scaling laws complete the proof.

\vspace{-0.8em}
\section{Proof of Lemma 2}
Using $  \delta_D = \frac{D_{\max}^2}{h^2}$, the vertical satellite-assisted FIM contribution can be rewritten as
\begin{equation}
\begin{aligned}
\bar K(D_{\max})\mu_z(D_{\max})
&=
\lambda_s\pi N_{\mathrm r}K_{\mathrm{eff}}^{\mathrm{sat}}
\Bigg[
b_r^2\ln(1+\delta_D) \\
&\hspace{-3em}
-4b_r\left(1-(1+\delta_D)^{-1/2}\right)
+\frac{\delta_D}{1+\delta_D}
\Bigg].
\end{aligned}
\end{equation}

For $D_{\max}\gg h$, $\delta_D\to\infty$, such that $(1+\delta_D)^{-1/2}=\mathcal{O}(\delta_D^{-1/2})$ and $\delta_D/(1+\delta_D)=1+\mathcal{O}(\delta_D^{-1})$. Therefore
\begin{equation}
\begin{aligned}
&\bar K(D_{\max})\mu_z(D_{\max})\\
&\hspace{-1em}\quad =
\lambda_s\pi N_{\mathrm r}
K_{\mathrm{eff}}^{\mathrm{sat}}
\Big[
b_r^2\ln(1+\delta_D)
-4b_r+1
+\mathcal O(\delta_D^{-1/2})
\Big].
\end{aligned}
\end{equation}

For $b_r\neq0$, the logarithmic term is dominant, which yields
\begin{equation}
\begin{aligned}
&\bar K(D_{\max})\mu_z(D_{\max})=
\lambda_s\pi N_{\mathrm r}
K_{\mathrm{eff}}^{\mathrm{sat}}
\mathcal O\left[
\ln\left(
1+\frac{D_{\max}^2}{h^2}
\right)
\right].
\end{aligned}
\end{equation}

For $b_r=0$, as $D_{\max}\rightarrow\infty$, the vertical FIM satisfies
\begin{equation}
\bar K(D_{\max})\mu_z(D_{\max})
=
\lambda_s\pi N_{\mathrm r}K_{\mathrm{eff}}^{\mathrm{sat}}
\frac{\delta_D}{1+\delta_D}
\longrightarrow
\lambda_s\pi N_{\mathrm r}K_{\mathrm{eff}}^{\mathrm{sat}},
\end{equation}

Hence, $\bar K(D_{\max})\mu_z(D_{\max}) = \lambda_s\pi N_{\mathrm r} K_{\mathrm{eff}}^{\mathrm{sat}} \mathcal O(1)$ as $ b_r=0$. This completes the proof.

\vspace{-0.8em}
\section{Proof of Proposition 2: }
We define  $\mathbf a_{o,k}(t_q)
    =
    \begin{bmatrix}
    a_{x,o,k}(t_q),
    a_{y,o,k}(t_q),
    -s_{o,k}(t_q)
    \end{bmatrix}^T$ with $s_{o,k}(t_q)
    =
    \sin\epsilon_{o,k}(H,t_q)
    \label{eq:walker_vertical_los_component}
$. The range gradient vector between target and radar can be expanded as
\begin{equation}
    \frac{\mathbf p_i-\mathbf r_n}
    {\left\|\mathbf p_i-\mathbf r_n\right\|}
    =
    \begin{bmatrix}
    -c_r\cos\theta_n\\
    -c_r\sin\theta_n\\
    b_r
    \end{bmatrix},
    \label{eq:receiver_los_components}
\end{equation}

Substituting the above gives
\begin{equation}
    \mathbf g^{\mathrm{sat}}_{o,k,n}(t_q)
    =
    \begin{bmatrix}
    a_{x,o,k}(t_q)-c_r\cos\theta_n\\
    a_{y,o,k}(t_q)-c_r\sin\theta_n\\
    b_r-s_{o,k}(t_q)
    \end{bmatrix}.
    \label{eq:walker_gradient_components}
\end{equation}

For the closed-form analysis, we adopt a large-system horizontally isotropic visible-observation approximation. For each visible satellite-time observation, the LoS unit vector $\mathbf a_{o,k}(t_q)$ in \eqref{eq:walker_satellite_los_local} can be written as
\begin{equation}
    \mathbf a_{o,k}(t_q)
    =
    \begin{bmatrix}
   - \sqrt{1-s_{o,k}^2(t_q)}
    \cos\alpha_{o,k}(t_q)\\
    -\sqrt{1-s_{o,k}^2(t_q)}
    \sin\alpha_{o,k}(t_q)\\
   - s_{o,k}(t_q)
    \end{bmatrix},
    \label{eq:walker_los_azimuth_decomposition}
\end{equation}

Therefore, under the horizontally isotropic visible-observation approximation, the satellite-assisted FIM averaged over the local horizontal azimuth can be written as
\begin{equation}
\overline{\mathbf J}^{\mathrm W}_{\mathrm{sat}}
=
\mathbb E_{\alpha_{o,k},\theta_n}
\left[
\mathbf J^{\mathrm W}_{\mathrm{sat},i}
\right]
\approx
\operatorname{diag}
\left(
\bar{\lambda}^{\mathrm{sat}}_{x},
\bar{\lambda}^{\mathrm{sat}}_{x},
\bar{\lambda}^{\mathrm{sat}}_{z}
\right),
\label{eq:isotropic_satellite_fim_diag}
\end{equation}
where
$\bar{\lambda}^{\mathrm{sat}}_{x}
=
\frac{N_r}{2}\left[(1+c_r^2)\mathcal I_0-\mathcal I_2\right]$
and
$\bar{\lambda}^{\mathrm{sat}}_{z}
=
N_r\left[\mathcal I_2-2b_r\mathcal I_1+b_r^2\mathcal I_0\right]$.

The weighted visible-observation moments are defined as
\begin{equation}
    \mathcal I_m
    =
    \sum_{q=1}^{Q}
    \sum_{o=1}^{N_{\mathrm{orb}}}
    \sum_{k=1}^{N_{\mathrm{sat}}}
    \chi_{o,k}(t_q)
    \omega^{\mathrm{sat}}_{o,k}(t_q)
    s_{o,k}^{m}(t_q),
    \quad m=0,1,2.
    \label{eq:isotropic_weighted_moments}
\end{equation}

Then effective weighted Walker moments are approximated by $ \overline{\mathcal I}_m
    \approx M_{\mathrm{vis}}\mu_m$ with $m=0,1,2$. Substituting these approximations into the above expressions completes the proof.
    
\vspace{-0.8em}
\section{Proof of Proposition 3: }
Under the horizontally isotropic visible-observation approximation, the second-order Taylor correction requires the following weighted visible moments
\begin{equation}
    \nu_r
    =
    \frac{
    \left(K_{\mathrm{eff}}^{\mathrm{sat}}\right)^2
    }
    {1-x_0}
    \int_{x_0}^{1}
    \frac{
    (R_sx-R_0)^r
    }
    {
    D(x)^{r/2+2}
    }
    dx,
    \hspace{+0.5em}
    r=0,1,2,3,4.
    \label{eq:second_order_visible_moments_integral}
\end{equation}

The corresponding second-order moments of the Walker hybrid FIM fluctuations are defined as
\begin{equation}
    \mathcal M_{ij}
    =
    \mathbb E_{\boldsymbol{\xi}}
    \left[
        \left(
            \Delta J_{\mathrm{hyb},ij}^{\mathrm W}
            (\boldsymbol{\xi})
        \right)^2
    \right],
    \qquad
    i,j\in\{x,y,z\}.
\end{equation}

The closed-form evaluation neglects the cross-observation covariances. Under this approximation, the diagonal second-order moments are
\begin{equation}
\begin{aligned}
    \mathcal M_{xx}
    &=
    M_{\mathrm{vis}}N_r^2
    \bigg[
    \frac{3}{8}
    \left(
    \nu_0-2\nu_2+\nu_4
    \right)
    +
    \frac{c_r^2}{2}
    \left(
    \nu_0-\nu_2
    \right)
    +
    \frac{c_r^4}{4}\nu_0
    \bigg]  \\
    &\quad
    -
    M_{\mathrm{vis}}N_r^2
    \left[
    \frac{(1+c_r^2)\mu_0-\mu_2}{2}
    \right]^2,
\end{aligned}
\label{eq:visible_Qxx}
\end{equation}
\begin{equation}
\begin{aligned}
    \mathcal M_{zz}
    &=
    M_{\mathrm{vis}}N_r^2
    \left[
    \nu_4
    -
    4b_r\nu_3
    +
    6b_r^2\nu_2
    -
    4b_r^3\nu_1
    +
    b_r^4\nu_0
    \right]  \\
    &\quad
    -
    M_{\mathrm{vis}}N_r^2
    \left[
    \mu_2
    -
    2b_r\mu_1
    +
    b_r^2\mu_0
    \right]^2.
\end{aligned}
\label{eq:visible_Qzz}
\end{equation}

The off-diagonal second-order moments are
\begin{equation}
     \mathcal M_{xy}
    =
    \frac{M_{\mathrm{vis}}N_r^2}{8}
    \left(
    \nu_0-2\nu_2+\nu_4
    \right),
    \label{eq:visible_Qxy}
\end{equation}
\begin{equation}
\begin{aligned}
    \mathcal M_{xz}
    &=
    \frac{M_{\mathrm{vis}}N_r^2}{2}
    \big[
    b_r^2\nu_0
    -
    2b_r\nu_1
    +
    (1-b_r^2)\nu_2  \\
    &\qquad
    +
    2b_r\nu_3
    -
    \nu_4
    \big].
\end{aligned}
\label{eq:visible_Qxz}
\end{equation}

Substituting $\nu_r$ into $\mathcal M_{ij}$ and subsequently into $\mathcal T^{\mathrm W}$ completes the proof.

\vspace{-1em}
%\appendices

\end{document}